\documentclass{scrartcl}

\usepackage[amsthm]{pamas}

\usepackage{todonotes}
\presetkeys{todonotes}{inline}{}

\usepackage{pgflibraryshapes} 
\usetikzlibrary{positioning} 
\usetikzlibrary{shapes.multipart} 
\usetikzlibrary{calc}

\newcommand{\lrs}{\ensuremath{\mathit{LRS}}\xspace}
\newcommand{\lrsin}{\ensuremath{\mathit{LRS_{IN}}}\xspace}
\newcommand{\lrsout}{\ensuremath{\mathit{LRS_{OUT}}}\xspace}

\newcommand{\nteq}{n_{\teq}}
\newcommand{\pos}{\ensuremath{\mathit{POS}}\xspace}

\newcommand{\dom}[1][]{\ifthenelse{\equal{#1}{}}{\overline{D}}{\overline(D)(#1)}}
\newcommand{\set}[1]{\{#1\}}

\newcommand{\simple}{simple\xspace}
\newcommand{\rootterm}{root\xspace}
\newcommand{\rootterms}{roots\xspace}
\newcommand{\rootsym}[1][]{\ensuremath{\lceil{#1}\rceil}}

\usepackage{pifont}
\newlength{\wordlength}

\newcommand{\mathwordbox}[3][c]{\settowidth{\wordlength}{$#3$}\makebox[\wordlength][#1]{$#2$}}

\def\nd{4em} 	
\def\ra{0.75em} 
\tikzset{node distance = \nd, 
		vertex/.style={circle,draw,minimum size=2*\ra, inner sep=1pt, fill=white}
		}

\newcommand{\tc}[0]{\ensuremath{\mathit{TC}}\xspace}
\newcommand{\uc}[0]{\ensuremath{\mathit{UC}}\xspace}

\newcommand{\mc}[0]{\ensuremath{\mathit{MC}}\xspace}
\newcommand{\ba}[0]{\ensuremath{\mathit{BA}}\xspace}
\newcommand{\bp}[0]{\ensuremath{\mathit{BP}}\xspace}
\newcommand{\teq}[0]{\ensuremath{\mathit{TEQ}}\xspace}

\newcommand{\me}[0]{\ensuremath{\mathit{ME}}\xspace}

\newcommand{\tcring}[0]{\ensuremath{\mathit{\mathring{TC}}}\xspace}

\DeclareMathAccent{\widehat}{\mathord}{largesymbols}{"62} 

\newcounter{remark}
\newenvironment{remark}[1][]{\refstepcounter{remark} \ifthenelse{\equal{#1}{}}{\medskip\noindent\textbf{Remark~\theremark. }}{\noindent \textbf{Remark~\theremark~(#1).}}}{\medskip}

\title{On the Structure of\\ Stable Tournament Solutions}
\author{
Felix Brandt\\
Technical University of Munich\\Germany
\and
Markus Brill\\
Oxford University\\United Kingdom
\and
Hans Georg Seedig\\
Technical University of Munich\\Germany
\and
Warut Suksompong\\
Stanford University\\USA
}

\date{}

\begin{document}
	
\maketitle

\begin{abstract}
A fundamental property of choice functions is stability, which, loosely speaking, prescribes that choice sets are invariant under adding and removing unchosen alternatives. 
We provide several structural insights that improve our understanding of stable choice functions. In particular, 
\emph{(i)}~we show that every stable choice function is generated by a unique simple choice function, which never excludes more than one alternative, \emph{(ii)}~we completely characterize which simple choice functions give rise to stable choice functions, and \emph{(iii)}~we prove a strong relationship between stability and a new property of tournament solutions called \emph{local reversal symmetry}. Based on these findings, we provide the first concrete tournament---consisting of 24 alternatives---in which the tournament equilibrium set fails to be stable. Furthermore, we prove that there is no more discriminating stable tournament solution than the bipartisan set and that the bipartisan set is the unique most discriminating tournament solution which satisfies standard properties proposed in the literature.
\end{abstract}

\section{Introduction}

Given a set of alternatives and binary non-transitive preferences over these alternatives, how can we consistently choose the ``best'' elements from any feasible subset of alternatives? This question has been studied in detail in the literature on tournament solutions \citep[see, \eg][]{Lasl97a,Hudr09a,Mose15a,BBH15a}. The lack of transitivity is typically attributed to the independence of pairwise comparisons as they arise in sports competitions, multi-criteria decision analysis, and preference aggregation.\footnote{Due to their generality, tournament solutions have also found applications in unrelated areas such as biology \citep{Schj22a,Land51a,Slat61a,AlLe11a}.} 
In particular, the pairwise majority relation of a profile of transitive individual preference relations often forms the basis of the study of tournament solutions. This is justified by a theorem due to \citet{McGa53a}, which shows that every tournament can be induced by some underlying preference profile. Many tournament solutions therefore correspond to well-known social choice functions such as Copeland's rule, Slater's rule, the Banks set, and the bipartisan set.

Over the years, many desirable properties of tournament solutions have been proposed. Some of these properties, so-called \emph{choice consistency conditions}, make no reference to the actual tournament but only relate choices from different subtournaments to each other. An important choice consistency condition, that goes under various names, requires that the choice set is invariant under the removal of unchosen alternatives. In conjunction with a dual condition on expanded feasible sets, this property is known as \emph{stability} \citep{BrHa11a}. Stability implies that choices are made in a robust and coherent way. Furthermore, stable choice functions can be rationalized by a preference relation on \emph{sets} of alternatives. 

Examples of stable tournament solutions are the \emph{top cycle}, \emph{the minimal covering set}, and the \emph{bipartisan set}. The latter is elegantly defined via the support of the unique mixed maximin strategies of the zero-sum game given by the tournament's skew-adjacency matrix.
Curiously, for some tournament solutions, including the \emph{tournament equilibrium set} and the \emph{minimal extending set}, proving or disproving stability turned out to be exceedingly difficult. As a matter of fact, whether the tournament equilibrium set satisfies stability was open for more than two decades before the existence of counterexamples with about $10^{136}$ alternatives was shown using the probabilistic method. 

\citet{Bran11b} systematically constructed stable tournament solutions by applying a well-defined operation to existing (non-stable) tournament solutions. \citeauthor{Bran11b}'s study was restricted to a particular class of generating tournament solutions, namely tournament solutions that can be defined via qualified subsets (such as the \emph{uncovered set} and the \emph{Banks set}). For any such generator, \citet{Bran11b} gave sufficient conditions for the resulting tournament solution to be stable.
Later, \citet{BHS15a} showed that for one particular generator, the Banks set, the sufficient conditions for stability are also necessary.

In this paper, we show that \emph{every} stable choice function is generated by a unique underlying simple choice function, which  never excludes more than one alternative.
We go on to prove a general characterization of stable tournament solutions that is not restricted to generators defined via qualified subsets. As a corollary, we obtain that the sufficient conditions for generators defined via qualified subsets are also necessary.  Finally, we prove a strong connection between stability and a new property of tournament solutions called \emph{local reversal symmetry}. Local reversal symmetry requires that an alternative is chosen if and only if it is unchosen when all its incident edges are inverted. This result allows us to settle two important problems in the theory of tournament solutions. 
We provide the first concrete tournament---consisting of 24 alternatives---in which the tournament equilibrium set violates stability. Secondly, we prove that there is no more discriminating stable tournament solution than the bipartisan set. We also axiomatically characterize the bipartisan set by only using properties that have been previously proposed in the literature. We believe that these results serve as a strong argument in favor of the  bipartisan set if choice consistency is desired.

\section{Stable Sets and Stable Choice Functions}
\label{sec:stability}

Let $U$ be a universal set of alternatives. Any finite non-empty subset of $U$ will be called a \emph{feasible set}. 
Before we analyze tournament solutions in \secref{sec:tsolutions}, we first consider a more general model of choice which does not impose any structure on feasible sets.
A \emph{choice function} is a function that maps every feasible set $A$ to a non-empty subset of $A$ called the \emph{choice set} of $A$. For two choice functions $S$ and $S'$, we write $S'\subseteq S$, and say that $S'$ is a \emph{refinement} of~$S$ and~$S$ a \emph{coarsening} of~$S'$, if $S'(A)\subseteq S(A)$ for all feasible sets~$A$. 
A choice function $S$ is called \emph{trivial} if $S(A)=A$ for all feasible sets $A$.

\citet{Bran11b} proposed a general method for refining a choice function~$S$ by defining minimal sets that satisfy internal and external stability criteria with respect to~$S$, similar to von-Neumann--Morgenstern stable sets in cooperative game theory.\footnote{This is a generalization of earlier work by \citet{Dutt88a}, who defined the minimal covering set as the unique minimal set that is internally and externally stable with respect to the uncovered set (see \secref{sec:tsolutions}).}

A subset of alternatives $X\subseteq A$ is called $S$-\emph{stable} within feasible set $A$ for choice function $S$ if it consists precisely of those alternatives that are chosen in the presence of all alternatives in $X$. Formally, $X$ is $S$-stable in $A$ if
\[
X=\{a\in A \colon a\in S(X\cup \{a\})\}\text. 
\]
Equivalently, $X$ is $S$-stable if and only if
\begin{gather}
S(X)=X \text{, and} \tag{internal stability}\\
a \notin S(X\cup\{a\}) \text{ for all }a\in A\setminus X\text. \tag{external stability}
\end{gather}
The intuition underlying this formulation is that there should be no reason to restrict the choice set by excluding some alternative from it (internal stability) and there should be an argument against each proposal to include an outside alternative into the choice set (external stability). 

An $S$-stable set is \emph{inclusion-minimal} (or simply \emph{minimal}) if it does not contain another $S$-stable set. $\widehat{S}(A)$ is defined as the union of all minimal $S$-stable sets in $A$. 
$\widehat{S}$ defines a choice function whenever every feasible set admits at least one $S$-stable set. In general, however, neither the existence of $S$-stable sets nor the uniqueness of minimal $S$-stable sets is guaranteed. We say that $\widehat{S}$ is \emph{well-defined} if every choice set admits exactly one minimal $S$-stable set. We can now define the central concept of this paper.

	\begin{definition}	
	A choice function $S$ is \emph{stable} if $\widehat{S}$ is well-defined and $S=\widehat{S}$.
         \end{definition}
Stability is connected to rationalizability and non-manipulability. In fact, every stable choice function can be rationalized via a preference relation on \emph{sets} of alternatives \citep{BrHa11a} and, in the context of social choice, stability and monotonicity imply strategyproofness with respect to Kelly's preference extension \citep{Bran11c}.

The following example illustrates the preceding definitions. Consider universe $U=\{a,b,c\}$ and choice function $S$ given by the table below (choices from singleton sets are trivial and therefore omitted). 
\[
	\begin{array}{ccc}
		X	&	S(X) & \widehat{S}(X)\\\midrule
		\set{a,b}	& \set{a} & \set{a}\\
		\set{b,c}	&\set{b} & \set{b}\\
		\set{a,c}	&\set{a} & \set{a}\\
		\set{a,b,c}	&\set{a,b,c} & \set{a}\\
	\end{array}
\]
The feasible set $\{a,b,c\}$ admits exactly two $S$-stable sets, $\{a,b,c\}$ itself and $\{a\}$. The latter holds because $S(\{a\})=\{a\}$ (internal stability) and $S(\{a,b\})=S(\{a,c\})=\{a\}$ (external stability).
All other feasible sets $X$ admit unique $S$-stable sets, which coincide with $S(X)$. Hence, $\widehat{S}$ is well-defined and given by the entries in the rightmost column of the table. Since $S\neq \widehat{S}$, $S$ fails to be stable. $\widehat{S}$, on the other hand, satisfies stability.

Choice functions are usually evaluated by checking whether they satisfy choice consistency conditions that relate choices from different feasible sets to each other. 
The following two properties, $\widehat\alpha$ and $\widehat\gamma$, are set-based variants of Sen's $\alpha$ and $\gamma$ \citep{Sen71a}.
$\widehat{\alpha}$ is a rather prominent choice-theoretic condition, also known as \citeauthor{Cher54a}'s \emph{postulate~$5^*$} \citep{Cher54a}, the \emph{strong superset property} \citep{Bord79a}, \emph{outcast}~\citep{AiAl95a}, and the \emph{attention filter axiom} \citep{MNO12a}.\footnote{We refer to \citet{Monj08a} for a more thorough discussion of the origins of this condition.}
		
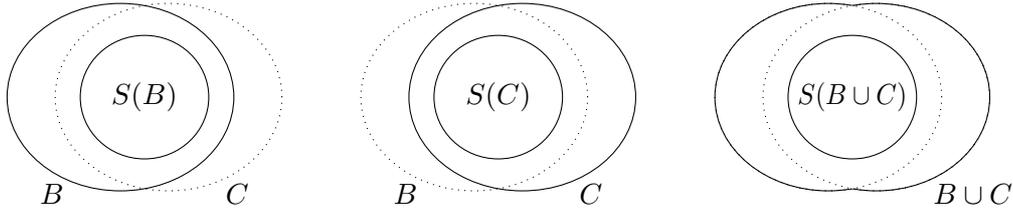
\begin{figure}[tb]
		\[
			\scalebox{1}{
			\begin{tikzpicture}[scale=1]
				 \draw (  0:9pt) node[ellipse,inner xsep=30pt,inner ysep=25pt,draw,dotted](B){} ++(220:57pt) node(){$B$}; 
				 \draw (180:9pt) node[ellipse,inner xsep=30pt,inner ysep=25pt,draw](C){}        ++(-40:57pt) node(){$C$};			 
				 \draw (  0:0pt) node[ellipse,inner xsep=17pt,inner ysep=11pt,draw](S){\mathwordbox{\scalebox{1}[1]{$S(B)$}}{}} ;			 
			\end{tikzpicture}
			}
			\qquad
			\scalebox{1}{
			\begin{tikzpicture}[scale=1]
				 \draw (  0:9pt) node[ellipse,inner xsep=30pt,inner ysep=25pt,draw](B){}       ++(220:57pt) node(){$B$}; 
				 \draw (180:9pt) node[ellipse,inner xsep=30pt,inner ysep=25pt,draw,dotted](C){}++(-40:57pt) node(){$C$};			 
				 \draw (  0:0pt) node[ellipse,inner xsep=17pt,inner ysep=11pt,draw](S){\mathwordbox{\scalebox{1}[1]{$S(C)$}}{}};			 
			\end{tikzpicture}
			}
			\qquad
			\scalebox{1}{
			\begin{tikzpicture}[scale=1]
				 \draw (  0:9pt) node[ellipse,inner xsep=30pt,inner ysep=25pt,draw,fill=white](B){} ; 
				 \draw (180:9pt) node[ellipse,inner xsep=30pt,inner ysep=25pt,draw,fill=white](C){}++(-40:57pt) node(){$\mathwordbox[l]{B\cup C}{C}$} ;
				 \draw (0:0pt)	 node[ellipse,draw=white,fill=white,inner xsep=30pt, inner ysep=24.23pt]{};
				 \draw (180:9pt) node[ellipse,inner xsep=30pt,inner ysep=25pt,draw,dotted](B){} ;
				 \draw (  0:9pt) node[ellipse,inner xsep=30pt,inner ysep=25pt,draw,dotted](C){} ;
				 \draw (  0:0pt) node[ellipse,inner xsep=17pt,inner ysep=11pt,draw](S){\mathwordbox{\scalebox{.9}[1]{$S(B\cup C)$}}{}} ;
			\end{tikzpicture}
			}
			\]
	\caption{Visualization of stability. 
		A stable choice function~$S$ chooses a set from both~$B$ (left) and~$C$ (middle) if and only if it chooses the same set from~$B\cup C$ (right). The direction from the left and middle diagrams to the right diagram corresponds to $\widehat{\gamma}$ while the converse direction corresponds to $\widehat{\alpha}$.
	}
	\label{fig:stability-illustration}
\end{figure}		

		\begin{definition}
		A choice function~$S$ satisfies $\widehat\alpha$ if for all feasible sets~$B$ and $C$,  
		\[
		\tag{$\widehat \alpha$}
			\text{
			    $S(B)\subseteq C\subseteq B$ implies $S(C)= S(B)$.
			}
		\]
		A choice function~$S$ satisfies $\widehat\gamma$ if for all feasible sets~$B$ and $C$,
		\[
		\tag{$\widehat\gamma$} 
			\text{ 
			    $S(B)=S(C)$ implies $S(B)=S(B\cup C)$.
			}
		\]
		\end{definition}

It has been shown that stability is equivalent to the conjunction of $\widehat\alpha$ and $\widehat\gamma$. 

\begin{theorem}[\citealp{BrHa11a}]\label{thm:BrHa}
A choice function is stable if and only if it satisfies $\widehat\alpha$ and $\widehat\gamma$.
\end{theorem}
Hence, a choice function~$S$ is stable if and only if for all feasible sets~$B$, $C$, and~$X$ with $X\subseteq B\cap C$,
		\[
			\text{
				$X=S(B)$ and $X=S(C)$
				\quad if and only if \quad 
				$X=S(B\cup C)$.
			}
		\]
Stability, $\widehat{\alpha}$, and $\widehat{\gamma}$ are illustrated in \figref{fig:stability-illustration}.

For a finer analysis, we split $\widehat{\alpha}$ and $\widehat{\gamma}$ into two conditions \citep[][Remark 1]{BrHa11a}.
\begin{definition}\label{def:greek-letter-properties}
	A choice function $S$ satisfies
	\begin{itemize}
		\item $\widehat{\alpha}_{_\subseteq}$ if for all $B,C$, it holds that $S(B)\subseteq C\subseteq B$ implies $S(C)\subseteq S(B)$,\footnote{$\widehat\alpha_{_\subseteq}$ has also been called the \emph{A\"izerman property} or the \emph{weak superset property} \citep[\eg][]{Lasl97a,Bran11b}.}
		\item $\widehat{\alpha}_{_\supseteq}$ if for all $B,C$, it holds that $S(B)\subseteq C\subseteq B$ implies $S(C)\supseteq S(B)$,
		\item $\widehat{\gamma}_{_\subseteq}$ if for all $B,C$, it holds that $S(B)=S(C)$ implies $S(B)\subseteq S(B\cup C)$, and
		\item $\widehat{\gamma}_{_\supseteq}$ if for all $B,C$, it holds that $S(B)=S(C)$ implies $S(B)\supseteq S(B\cup C)$.
	\end{itemize}
\end{definition}

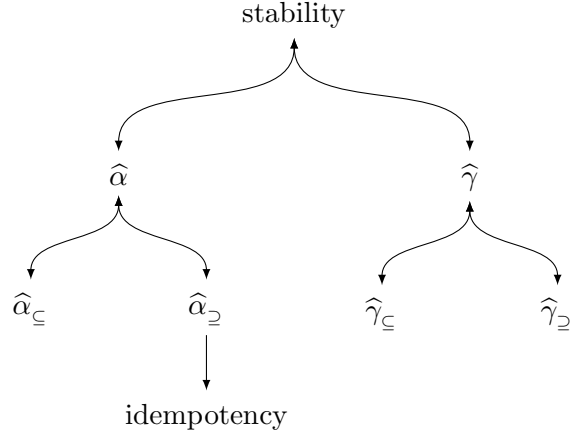
\begin{figure}[tb]
	\centering
		\begin{tikzpicture}[node distance=3em]
		\tikzstyle{pfeil}=[latex-latex, shorten >=1pt,draw]
		\tikzstyle{onlytext}=[]

		\node[onlytext] (SS) at (0,0) {stability};
		\node[onlytext] (ahat) [below=of SS,yshift=-1em,xshift=-6em] 			{$\widehat\alpha$};
		\node[onlytext] (ahat-incl) [below=of ahat,xshift=-3em] 	{$\widehat\alpha_{_\subseteq}$};
		\node[onlytext] (ahat-sup) [below=of ahat,xshift=3em] 	{$\widehat\alpha_{_\supseteq}$};
		\node[onlytext] (ghat) [below=of SS,yshift=-1em,xshift=6em] 			{$\widehat\gamma$};
		\node[onlytext] (ghat-incl) [below=of ghat,xshift=-3em] 	{$\widehat\gamma_{_\subseteq}$};
		\node[onlytext] (ghat-sup) [below=of ghat,xshift=3em] 	{$\widehat\gamma_{_\supseteq}$};
		\node[onlytext,node distance=2em] (ide) [below=of ahat-sup]	{idempotency};
		
		\draw[pfeil] (SS) 	to[out=270,in=90] (ahat); 
		\draw[pfeil] (ahat)	to[out=270,in=90] (ahat-incl);
		\draw[pfeil] (ahat) to[out=270,in=90] (ahat-sup);
		\draw[pfeil] (SS) 	to[out=270,in=90] (ghat);
		\draw[pfeil] (ghat) to[out=270,in=90] (ghat-incl);
		\draw[pfeil] (ghat) to[out=270,in=90] (ghat-sup);		
		\draw[pfeil,-latex] (ahat-sup) to[out=270,in=90] (ide);
	\end{tikzpicture}
	\caption{Logical relationships between choice-theoretic properties.}
	\label{fig:stability-properties}
\end{figure}

Obviously, for any choice function $S$ we have 
\begin{align*}
	S \text{ satisfies }\widehat{\alpha} \quad&\text{ if and only if }\quad S \text{ satisfies }\widehat{\alpha}_{_\subseteq} \text{ and } \widehat{\alpha}_{_\supseteq}  \text{, and} \\
	S \text{ satisfies }\widehat{\gamma} \quad&\text{ if and only if }\quad S \text{ satisfies } \widehat{\gamma}_{_\subseteq} \text{ and } \widehat{\gamma}_{_\supseteq}. 
\end{align*}
%
%
A choice function is \emph{idempotent} if the choice set is invariant under repeated application of the choice function, 
\ie $S(S(A))=S(A)$ for all feasible sets~$A$.
It is easily seen that $\widehat{\alpha}_{_\supseteq}$ is stronger than idempotency since $S(S(A))\supseteq S(A)$ implies $S(S(A))=S(A)$.

\figref{fig:stability-properties} shows the logical relationships between stability and its weakenings.

\section{Generators of Stable Choice Functions}
\label{sec:generators}

We say that a choice function $S'$ \emph{generates} a stable choice function $S$ if $S=\widehat{S'}$. Understanding stable choice functions can be reduced to understanding their generators. 
It turns out that important generators of stable choice functions are \emph{simple} choice functions, \ie choice functions $S'$ with $|S'(A)|\ge |A|-1$ for all $A$. In fact, every stable choice function $S$ is generated by a unique \simple choice function.
To this end, we define the \emph{\rootterm} of a choice function $S$ as 
\[
\rootsym[S](A)=\begin{cases}S(A) &\mbox{if } |S(A)|=|A|-1\text{,} \\ A &\mbox{otherwise.}\end{cases}
\]
Not only does $\rootsym[S]$ generate $S$, but any choice function sandwiched between $S$ and $\rootsym[S]$ is a generator of $S$.

\begin{theorem}
\label{thm:generatorS}
Let $S$ and $S'$ be choice functions such that $S$ is stable and $S\subseteq S'\subseteq \rootsym[S]$. Then, $\widehat{S'}$ is well-defined and $\widehat{S'}=S$. 
In particular, $S$ is generated by the simple choice function $\rootsym[S]$.
\end{theorem}

\begin{proof}
We first show that any $S$-stable set is also $S'$-stable. Suppose that a set $X\subseteq A$ is $S$-stable in $A$. Then $S(X)=X$, and $a\not\in S(X\cup\{a\})$ for all $a\in A\backslash X$. Since $S$ satisfies~$\widehat{\alpha}$, we have $S(X\cup\{a\})=X$ and therefore $\rootsym[S](X\cup\{a\})=X$ for all $a\in A\backslash X$. Using the inclusion relationship $S\subseteq S'\subseteq \rootsym[S]$, we find that $S'(X)=X$ and $S'(X\cup\{a\})=X$ for all $a\in A\backslash X$. Hence, $X$ is $S'$-stable in $A$.

Next, we show that every $S'$-stable set contains an $S$-stable set. Suppose that a set $X\subseteq A$ is $S'$-stable in $A$. Then $S'(X)=X$ and $a\not\in S'(X\cup\{a\})$ for all $a\in A\backslash X$. Using the relation $S\subseteq S'$, we find that $a\not\in S(X\cup\{a\})$ for all $a\in A\backslash X$. We will show that $S(X)\subseteq X$ is $S$-stable in $A$. Since $S$ satisfies $\widehat{\alpha}$, we have $S(S(X))=S(X)$ and $S(X\cup\{a\})=S(X)$ for all $a\in A\backslash X$. It remains to show that $b\not\in S(S(X)\cup\{b\})$ for all $b\in A\backslash S(X)$. If $b\in A\backslash X$, we already have that $S(X\cup\{b\})=S(X)$ and therefore $S(S(X)\cup\{b\})=S(X)$ by $\widehat{\alpha}$. Otherwise, if $b\in X\backslash S(X)$, $\widehat{\alpha}$ again implies that $S(S(X)\cup\{b\})=S(X)$.

Since $S$ is stable, for any feasible set $A$ there exists a unique minimal $S$-stable set in $A$, which is given by $S(A)=\widehat{S}(A)$. From what we have shown, this set is also $S'$-stable, and moreover any $S'$-stable set contains an $S$-stable set which in turn contains $S(A)$. Hence $S(A)$ is also the unique minimal $S'$-stable set in $A$. This implies that $\widehat{S'}$ is well-defined and $\widehat{S'}=\widehat{S}=S$.
\end{proof}

\thmref{thm:generatorS} entails that in order to understand stable choice functions, we only need to understand the circumstances under which a single alternative is discarded.\footnote{Together with \thmref{thm:ShatdirectedMSSP}, \thmref{thm:generatorS} also implies that, for any stable tournament solution $S$, $\rootsym[S]$ is a coarsest generator of $S$. When only considering generators that satisfy $\widehat{\alpha}_{_\subseteq}$, $\rootsym[S]$ is also \emph{the} coarsest generator of $S$. In addition, since simple choice functions trivially satisfy $\widehat{\alpha}_{_\subseteq}$, the two theorems imply that $\rootsym[S]$ is the unique simple choice function generating $S$.}

An important question is which simple choice functions are \rootterms of stable choice functions. It follows from the definition of \rootterm functions that any \rootterm of a stable choice function needs to satisfy $\widehat{\alpha}$. This condition is, however, not sufficient as it is easy to construct a simple choice function $S$ that satisfies $\widehat{\alpha}$ such that $\widehat{S}$ violates $\widehat{\alpha}$. Nevertheless, the theorem implies that the number of stable choice functions can be bounded by counting the number of simple choice functions that satisfy $\widehat{\alpha}$. The number of simple choice functions for a universe of size $n\ge 2$ is only $\prod_{i=2}^n (i+1)^{\binom{n}{i}}$, compared to $\prod_{i=2}^n (2^i-1)^{\binom{n}{i}}$ for arbitrary choice functions.

In order to give a complete characterization of choice functions that generate stable choice functions, we need to introduce a new property.
A choice function $S$ satisfies local $\widehat{\alpha}$ if minimal $S$-stable sets are invariant under removing outside alternatives.\footnote{It can be checked that we obtain an equivalent condition even if we require that \emph{all} outside alternatives have to be removed. When defining local $\widehat{\alpha}$ in this way, it can be interpreted as some form of transitivity of stability: stable sets of minimally stable sets are also stable within the original feasible set \citep[cf.][Lem.~3]{Bran11b}.}

\begin{definition}
A choice function $S$ satisfies \emph{local} $\widehat{\alpha}$ if for any sets $X\subseteq Y\subseteq Z$ such that $X$ is minimally $S$-stable in $Z$, we have that $X$ is also minimally $S$-stable in $Y$.
\end{definition}

Recall that a choice function $S$ satisfies $\widehat{\alpha}_{_\subseteq}$ if for any sets $A,B$ such that $S(A)\subseteq B\subseteq A$, we have $S(B)\subseteq S(A)$. In particular, every simple choice function satisfies $\widehat{\alpha}_{_\subseteq}$. We will provide a characterization of choice functions $S$ satisfying $\widehat{\alpha}_{_\subseteq}$ such that $\widehat{S}$ is stable. First we need the following (known) lemma.

\begin{lemma}[\citealp{BrHa11a}]
\label{lemma:Shatgammahat}
Let $S$ be a choice function such that $\widehat{S}$ is well-defined. Then $\widehat{S}$ satisfies $\widehat{\gamma}$.
\end{lemma}

\begin{theorem}
\label{thm:ShatdirectedMSSP}
Let $S$ be a choice function satisfying $\widehat{\alpha}_{_\subseteq}$. Then $\widehat{S}$ is stable if and only if $\widehat{S}$ is well-defined and $S$ satisfies local $\widehat{\alpha}$.
\end{theorem}

\begin{proof}
For the direction from right to left, suppose that $\widehat{S}$ is well-defined and $S$ satisfies local $\widehat{\alpha}$. Then Lemma \ref{lemma:Shatgammahat} implies that $\widehat{S}$ satisfies $\widehat{\gamma}$. Moreover, it follows directly from local $\widehat{\alpha}$ and the fact that $\widehat{S}$ is well-defined that $\widehat{S}$ satisfies $\widehat{\alpha}$. Hence, $\widehat{S}$ is stable.

For the converse direction, suppose that $\widehat{S}$ is stable. We first show that $\widehat{S}$ is well-defined. 
Every feasible set $A$ contains at least one $S$-stable set because otherwise $\widehat{S}$ is not a choice function.
Next, suppose for contradiction that there exists a feasible set that contains two distinct minimal $S$-stable sets. Consider such a feasible set $A$ of minimum size, and pick any two distinct minimal $S$-stable sets in $A$, which we denote by $B$ and $C$. If $|B\backslash C|=|C\backslash B|=1$, then $\widehat{\alpha}_{_\subseteq}$ implies $S(B\cup C)=B=C$, a contradiction. Otherwise, assume without loss of generality that $|C\backslash B|\geq 2$, and pick $x,y\in C\backslash B$ with $x\neq y$. Then $A\backslash\{x\}$ contains a unique minimal $S$-stable set. As $B$ is also $S$-stable in $A\backslash\{x\}$, it follows that $\widehat{S}(A\backslash\{x\})\subseteq B$. Since $\widehat{S}$ satisfies $\widehat{\alpha}$, we have $\widehat{S}(A\backslash\{x\})=\widehat{S}(A\backslash\{x,y\})$. Similarly, we have $\widehat{S}(A\backslash\{y\})=\widehat{S}(A\backslash\{x,y\})$. But then $\widehat{\gamma}$ implies that $\widehat{S}(A)=\widehat{S}(A\backslash\{x,y\})\subseteq A$, which contradicts the fact that $C$ is minimal $S$-stable in $A$.

We now show that $S$ satisfies local $\widehat{\alpha}$. Since $\widehat{S}$ is well-defined, minimal $S$-stable sets are unique and given by $\widehat{S}$. Since $\widehat{S}$ satisfies $\widehat{\alpha}$, minimal $S$-stable sets are invariant under deleting outside alternatives. Hence, $S$ satisfies local $\widehat{\alpha}$, as desired.
\end{proof}

\begin{remark}
Theorem \ref{thm:ShatdirectedMSSP} does not hold without the condition that $S$ satisfies $\widehat{\alpha}_{_\subseteq}$. To this end, let $U=\{a,b,c\}$, $S(\{a,b,c\})=\{b\}$, and $S(X)=X$ for all other feasible sets $X$. Then both $\{a,b\}$ and $\{b,c\}$ are minimally $S$-stable in $\{a,b,c\}$, implying that $\widehat{S}$ is not well-defined. On the other hand, $\widehat{S}$ is trivial and therefore also stable. This example also shows that a generator of a stable choice function needs not be sandwiched between the choice function and its root.
\end{remark}

Combining Theorem \ref{thm:ShatdirectedMSSP} with \thmref{thm:BrHa}, we obtain the following characterization.

\begin{corollary}
Let $S$ be a choice function satisfying $\widehat{\alpha}_{_\subseteq}$. Then, 
\begin{align*}
\widehat{S} \text{ is stable}  &\text{ if and only if } \widehat{\widehat{S}} \text{ is well-defined and } \widehat{\widehat{S}}=\widehat{S} \\
&\text{ if and only if } \widehat{S} \text{ satisfies } \widehat{\alpha} \text{ and } \widehat{\gamma} \\
&\text{ if and only if } \widehat{S} \text{ is well-defined and } S \text{ satisfies local $\widehat{\alpha}$}.
\end{align*}
\end{corollary}

Since simple choice functions trivially satisfy $\widehat{\alpha}_{_\subseteq}$, this corollary completely characterizes which simple choice functions generate stable choice functions.

\section{Tournament Solutions}
\label{sec:tsolutions}

We now turn to the important special case of choice functions whose output depends on a binary relation.

\subsection{Preliminaries}
\label{sec:prelims}

A \emph{tournament $T$} is a pair $(A,{\succ})$, where $A$ is a feasible set and~$\succ$ is a connex and asymmetric (and thus irreflexive) binary relation on $A$, usually referred to as the \emph{dominance relation}. 
	Intuitively, $a\succ b$ signifies that alternative~$a$ is preferable to alternative~$b$. The dominance relation can be extended to sets of alternatives by writing $X\succ Y$ when $a\succ b$ for all $a\in X$ and $b\in Y$. 

	For a tournament $T=(A,\succ)$ and an alternative $a\in A$, 
	we denote by \[\dom(a)=\{\,x\in A\mid x \succ a\,\}\] the \emph{dominators} of~$a$
        and by \[D(a)=\{\,x\in A\mid a \succ x\,\}\] the \emph{dominion} of~$a$. When varying the tournament, we will refer to $\dom_{T'}(a)$ and $D_{T'}(a)$ for some tournament $T'=(A',\succ')$.
        An alternative $a$ is said to \emph{cover} another alternative $b$ if $D(b)\subseteq D(a)$. 
        It is said to be a \emph{Condorcet winner} if it dominates all other alternatives, and a \emph{Condorcet loser} if it is dominated by all other alternatives. 
	The order of a tournament $T=(A,\succ)$ is denoted by $|T|=|A|$. A tournament is \emph{regular} if the dominator set and the dominion set of each alternative are of the same size, \ie for all $a\in A$ we have $|D(a)|=|\dom(a)|$. It is easily seen that regular tournaments are always of odd order.
	
	A \emph{tournament solution} is a function that maps a tournament to a nonempty subset of its alternatives.
	We assume that tournament solutions are invariant under tournament isomorphisms. For every fixed tournament, a tournament solution yields a choice function. A tournament solution is \emph{trivial} if it returns all alternatives of every tournament.
	
	Three common tournament solutions are the top cycle, the uncovered set, and the Banks set. For a given tournament $(A,{\succ})$, the \emph{top cycle} (\tc) is the (unique) smallest set $B\subseteq A$ such that $B\succ A\setminus B$, the \emph{uncovered set} (\uc) contains all alternatives that are not covered by another alternative, and the \emph{Banks set} (\ba) contains all alternatives that are Condorcet winners of inclusion-maximal transitive subtournaments. 
	
	For two tournament solutions $S$ and $S'$, we write $S'\subseteq S$, and say that $S'$ is a \emph{refinement} of~$S$ and~$S$ a \emph{coarsening} of~$S'$, if $S'(T)\subseteq S(T)$ for all tournaments~$T$. The following inclusions are well-known:
	\[
	\ba \subseteq \uc \subseteq \tc\text.
	\]
	
	To simplify notation, we will often identify a (sub)tournament by its set of alternatives when the dominance relation is clear from the context. For example, for a tournament solution $S$ and a subset of alternatives $X\subseteq A$ in a tournament $T=(A,\succ)$ we will write $S(X)$ for $S(T|_{X})$. 

The definitions of stability and other choice consistency conditions directly carry over from choice functions to tournament solutions by requiring that the given condition should be satisfied by every choice function induced by the tournament solution and a tournament.
We additionally consider the following desirable properties of tournament solutions, all of which are standard conditions in the literature.

Monotonicity requires that a chosen alternative will still be chosen when its dominion is enlarged, while leaving everything else unchanged.

	\begin{definition}
	A tournament solution is \emph{monotonic} if for all $T=(A,{\succ})$, $T'=(A,{\succ'})$, $a \in A$ such that ${\succ}_{A\setminus\{a\}} = {\succ'}_{A\setminus\{a\}}$ and for all $b\in A\setminus \{a\}$, $a\succ' b$ whenever $a \succ b$, \[a\in S(T) \quad\text{implies}\quad a\in S(T')\text.\]
        \end{definition}

Regularity requires that all alternatives are chosen from regular tournaments.

	\begin{definition}
	 A tournament solution is \emph{regular} if $S(T)=A$ for all regular tournaments $T=(A,\succ)$.
        \end{definition}
Even though regularity is often considered in the context of tournament solutions, it does not possess the normative appeal of other conditions. 
		
Finally, we consider a structural invariance property that is based on components of similar alternatives and, loosely speaking, requires that a tournament solution chooses the ``best'' alternatives from the ``best'' components.
A \emph{component} is a nonempty subset of alternatives $B\subseteq A$ that bear the same relationship to any alternative not in the set, i.e., for all $a\in A\backslash B$, either $B\succ\{a\}$ or $\{a\}\succ B$. A \emph{decomposition} of $T$ is a partition of $A$ into components.

For a given tournament $\tilde{T}$, a new tournament $T$ can be constructed by replacing each alternative with a component. Let $B_1,\dots,B_k$ be pairwise disjoint sets of alternatives and consider tournaments $T_1=(B_1,\succ_1),\dots,T_k=(B_k,\succ_k)$, and $\tilde{T} = (\{1,\dots,k\}, \tilde{\succ})$. The \emph{product} of $T_1,\dots,T_k$ with respect to $\tilde{T}$, denoted by $\prod(\tilde{T},T_1,\dots,T_k)$, is the tournament $T=(A,\succ)$ such that $A=\bigcup_{i=1}^kB_i$ and for all
$b_1\in B_i,b_2\in B_j$,
\[b_1 \succ b_2 \text{ \hspace{0.1cm}  if and only if  \hspace{0.1cm} } i = j \text{ and } b_1\succ_i b_2, \text{ or } i \neq j \text{ and } i \mathrel{\tilde{\succ}} j.\]
Here, $\tilde{T}$ is called the \emph{summary} of T with respect to the above decomposition.

	\begin{definition}
	A tournament solution is \emph{composition-consistent} if for all tournaments $T,T_1,\dots,T_k$ and $\tilde{T}$ such that $T=\prod(\tilde{T},T_1,\dots,T_k)$,
\[S(T)=\bigcup_{i\in S(\tilde{T})}S(T_i).\]
        \end{definition}

All of the three tournament solutions we briefly introduced above satisfy monotonicity. \tc and \uc are regular, \uc and \ba are composition-consistent, and only \tc is stable.
	For more thorough treatments of tournament solutions, see \citet{Lasl97a} and \citet{BBH15a}.

\subsection{The Bipartisan set and the Tournament Equilibrium Set}
\label{sec:bpandteq}

We now define two tournament solutions that are central to this paper. The first one, the bipartisan set, generalizes the notion of a Condorcet winner to probability distributions over alternatives.  
The \emph{skew-adjacency matrix} $G(T)=(g_{ab})_{a,b\in A}$ of a tournament $T$ is defined by letting
\[
g_{ab} = \begin{cases}
1 & \text{if $a\succ b$}\\
-1 & \text{if $b\succ a$}\\
0 & \text{if $a=b$.}
\end{cases}
\]
The skew-adjacency matrix can be interpreted as a symmetric zero-sum game in which there are two players, one choosing rows and the other choosing columns, and in which the matrix entries are the payoffs of the row player. \citet{LLL93b} and \citet{FiRy95a} have shown independently that every such game admits a unique mixed Nash equilibrium, which moreover is symmetric. Let $p_T\in \Delta(A)$ denote the mixed strategy played by both players in equilibrium. Then, $p_T$ is the unique probability distribution such that
\[
	\sum_{a,b\in A} p_T(a)q(b)g_{ab}\ge 0 \quad\text{ for all }q\in\Delta(A)\text{.}
\] 
In other words, there is no other probability distribution that is more likely to yield a better alternative than $p_T$.
\citet{LLL93b} defined the bipartisan set~$\bp(T)$ of~$T$ as the support of $p_T$.\footnote{The probability distribution $p_T$ was independently analyzed by \citet{Krew65a}, \citet{Fish84a}, \citet{FeMa92a}, and others. An axiomatic characterization in the context of social choice was recently given by \citet{Bran13a}.}
\begin{definition}
The \emph{bipartisan set} ($\bp$) of a given tournament $T=(A,\succ)$ is defined as
\[ \bp(T) = \{ a\in A \mid p_T(a)>0 \}\text{.}\]
\end{definition}
\bp satisfies stability, monotonicity, regularity, and composition-consistency. Moreover, $\bp\subseteq \uc$ and $\bp$ can be computed in polynomial time by solving a linear feasibility problem.

The next tournament solution, the tournament equilibrium set, was defined by \citet{Schw90a}.
Given a tournament $T=(A,\succ)$ and a tournament solution $S$, a nonempty subset of alternatives $X\subseteq A$ is called $S$-\emph{retentive} if $S(\dom(x))\subseteq X$ for all $x \in X$ such that $\dom(x)\neq \emptyset$.

\begin{definition}
 The \emph{tournament equilibrium set} ($\teq$) of a given tournament $T=(A,\succ)$ is defined recursively as the union of all inclusion-minimal $\teq$-retentive sets in $T$.
\end{definition}

This is a proper recursive definition because the cardinality of the set of dominators of an alternative in a particular set is always smaller than the cardinality of the set itself. \bp and \teq coincide on all tournaments of order 5 or less \citep{BDS13a}.\footnote{It is open whether there are tournaments in which \bp and \teq are disjoint.}

\citet{Schw90a} showed that $\teq\subseteq \ba$ and conjectured that every tournament contains a \emph{unique} inclusion-minimal $\teq$-retentive set, which was later shown to be equivalent to $\teq$ satisfying any one of a number of desirable properties for tournament solutions including stability and monotonicity \citep{LLL93a,Houy09a,Houy09b,BBFH11a,Bran11b,BrHa11a,Bran11c}.
This conjecture was disproved by \citet{BCK+11a}, who have non-constructively shown the existence of a counterexample with about $10^{136}$ alternatives using the probabilistic method. Since it was shown that $\teq$ satisfies the above mentioned desirable properties for all tournaments that are smaller than the smallest counterexample to Schwartz's conjecture, the search for smaller counterexamples remains an important problem. In fact, the counterexample found by \citet{BCK+11a} is so large that it has no practical consequences whatsoever for $\teq$. Apart from concrete counterexamples, there is ongoing interest in why and under which circumstances $\teq$ and a related tournament solution called the \emph{minimal extending set} $\me=\widehat{\ba}$ violate stability \citep{MSY15a,BHS15a,Yang16a}.

Computing the tournament equilibrium set of a given tournament was shown to be NP-hard and consequently there does not exist an efficient algorithm for this problem unless P equals NP \citep{BFHM09a}.

\section{Stable Tournament Solutions and Their Generators}

Tournament solutions comprise an important subclass of choice functions. In this section, we examine the consequences of the findings from Sections~\ref{sec:stability} and \ref{sec:generators}, in particular Theorems~\ref{thm:BrHa}, \ref{thm:generatorS}, and \ref{thm:ShatdirectedMSSP}, for tournament solutions.

Stability is a rather demanding property which is satisfied by only a few tournament solutions.
Three well-known tournament solutions that satisfy stability are the top cycle \tc, the minimal covering set \mc defined by $\mc=\widehat{\uc}$, and the bipartisan set \bp, which is a refinement of \mc.
By virtue of \thmref{thm:generatorS}, any stable tournament solution is generated by its \rootterm $\rootsym[S]$. 
For example, $\rootsym[\tc]$ is a tournament solution that excludes an alternative if and only if it is the only alternative not contained in the top cycle (and hence a Condorcet loser). Similarly, one can obtain the \rootterms of other stable tournament solutions such as \mc and \bp. In some cases, the generator that is typically considered for a stable tournament solution is different from its \rootterm; for example, $\mc$ is traditionally generated by \uc, a refinement of $\rootsym[\mc]$.
Since tournament solutions are invariant under tournament isomorphisms, a simple tournament solution may only exclude 
an alternative $a$ if any automorphism of $T$ maps $a$ to itself.
Note that if a tournament solution $S$ is stable, $\rootsym[S]$ is different from $S$ unless $S$ is the trivial tournament solution. 

It follows from \thmref{thm:BrHa} that stable tournament solutions satisfy both $\widehat{\alpha}$ and $\widehat{\gamma}$. 
It can be shown that $\widehat{\alpha}$ and $\widehat{\gamma}$ are independent from each other even in the context of tournament solutions.

\begin{remark}\label{rem:alphagamma}
There are tournament solutions that satisfy only one of $\widehat{\alpha}$ and $\widehat{\gamma}$. Examples are given in Appendix~\ref{app:alphagamma}.
\end{remark}

We have shown in \thmref{thm:generatorS} that stable tournament solutions are generated by unique simple tournament solutions.
If we furthermore restrict our attention to \emph{monotonic} stable tournament solutions, the following theorem shows that we only need to consider \rootterm solutions that are monotonic.


\begin{theorem}
\label{thm:mon}
A stable tournament solution $S$ is monotonic if and only if $\rootsym[S]$ is monotonic.
\end{theorem}
\begin{proof}
First, note that monotonicity is equivalent to requiring that unchosen alternatives remain unchosen when they are weakened.
Now, for the direction from left to right, suppose that $S$ is monotonic. Let $T=(A,{\succ})$, $B=\rootsym[S](T)$, and $a\in A\setminus B$. Since $\rootsym[S]$ is simple, we have $\rootsym[S](T)=B\backslash\{a\}$, and therefore $S(T)=B\backslash\{a\}$. Using the fact that $S$ is stable and thus satisfies $\widehat{\alpha}$, we find that $S(T|_{B\backslash\{a\}})=B\backslash\{a\}$. Let $T'$ be a tournament obtained by weakening $a$ with respect to some alternative in $B$. Monotonicity of $S$ entails that $a\not\in S(T')$. Since $T|_{B\backslash\{a\}}=T'|_{B\backslash\{a\}}$, we have $S(T'|_{B\backslash\{a\}})=B\backslash\{a\}$, and $\widehat{\alpha}$ implies that $S(T')=B\backslash\{a\}$ and $\rootsym[S](T')=B\backslash\{a\}$ as well. This means that $a$ remains unchosen by $\rootsym[S]$ in $T'$, as desired.

The converse direction even holds for all generators of $S$ \citep[see][Prop.~5]{Bran11b}.
\end{proof}

Analogous results do \emph{not} hold for composition-consistency or regularity.

Theorem~\ref{thm:ShatdirectedMSSP} characterizes stable choice functions $\widehat{S}$ using well-definedness of $\widehat{S}$ and local $\widehat{\alpha}$ of $S$. These two properties are independent from each other (and therefore necessary for the characterization) even in the context of tournament solutions. 

\begin{remark}\label{rem:localalpha}
There is a tournament solution $S$ that satisfies local $\widehat{\alpha}$, but $\widehat{S}$ violates~$\widehat{\alpha}$. There is a tournament solution $S$ for which $\widehat{S}$ is well-defined, but $\widehat{S}$ is not stable. 
Examples are given in Appendix~\ref{app:localalpha}.
\end{remark}

\thmref{thm:ShatdirectedMSSP} generalizes previous statements about stable tournament solutions.
\citet{Bran11b} studied a particular class of generators defined via qualified subsets and shows the direction from right to left of \thmref{thm:ShatdirectedMSSP} for these generators \citep[][Thm.~4]{Bran11b}.\footnote{\citeauthor{Bran11b}'s proof relies on a lemma that essentially showed that the generators he considers always satisfy local $\widehat{\alpha}$.}
Later, \citet{BHS15a} proved \thmref{thm:ShatdirectedMSSP} for one particular generator \ba \citep[][Cor.~2]{BHS15a}.

\section{Local Reversal Symmetry}

In this section, we introduce a new property of tournament solutions called local reversal symmetry (\lrs).\footnote{The name of this axiom is inspired by a social choice criterion called \emph{reversal symmetry}. Reversal symmetry prescribes that a uniquely chosen alternative has to be unchosen when the preferences of all voters are reversed \citep{SaBa03a}. A stronger axiom, called \emph{ballot reversal symmetry}, which demands that the choice set is inverted when all preferences are reversed was recently introduced by \citet{DHLP+14a}.}
While intuitive by itself, \lrs is strongly connected to stability and can be leveraged to disprove that $\teq$ is stable and to prove that no refinement of $\bp$ is stable.


For a tournament $T$, let $T^a$ be the tournament whose dominance relation is \emph{locally reversed} at alternative $a$, \ie $T^a=(A,\succ^a)$ with
\[
	i \succ^a j \quad \text{if and only if} \quad 
				(i \succ j \text{ and } a \notin \{i,j\}) \text{ or }
				(j \succ i \text{ and } a \in \{i,j\}).
\]
The effect of local reversals is illustrated in \figref{fig:lrs-def}. Note that $T=\left(T^a\right)^a$ and $\left(T^a\right)^b=\left(T^b\right)^a$ for all alternatives $a$ and $b$.
	
	\begin{figure}[htb]
		\centering
		\begin{tikzpicture}[]
			\node (a) at (0,0) {$a$};
			\node (b) [right of=a] {$b$};
			\node (c) [below of=b] {$c$};
			\node (d) [left of=c] {$d$};
			\foreach \x / \y in {a/b,a/c,b/c,b/d,c/d,d/a}{
			\draw[-latex] (\x) to (\y);
			};
			\node (caption) [node distance = 0.5*\nd,below of=d, xshift=0.5*\nd] {$T$};
		\end{tikzpicture}
		\qquad\qquad
		\begin{tikzpicture}[]
			\node (a) at (0,0) {$a$};
			\node (b) [right of=a] {$b$};
			\node (c) [below of=b] {$c$};
			\node (d) [left of=c] {$d$};
			\foreach \x / \y in {b/a,c/a,b/c,b/d,c/d,a/d}{
			\draw[-latex] (\x) to (\y);
			};
			\node (caption) [node distance = 0.5*\nd,below of=d, xshift=0.5*\nd] {$T^a$};
		\end{tikzpicture}
		\caption{Local reversals of tournament $T$ at alternative $a$ result in $T^a$. $\bp(T)=\teq(T)=\{a,b,d\}$ and $\bp(T^a)=\teq(T^a)=\{b\}$.}
		\label{fig:lrs-def}
	\end{figure}
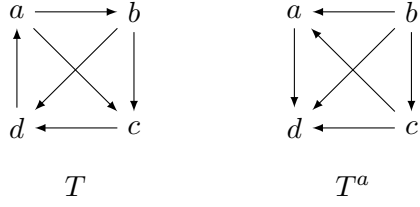

\begin{definition}
A tournament solution $S$ satisfies \emph{local reversal symmetry} (\lrs) if for all tournaments $T$ and alternatives $a$,
\[
			a\in S(T) \text{ if and only if } a \notin S(T^a).
\]
\end{definition}

\lrs can be naturally split into two properties, \lrsin and \lrsout. \lrsin corresponds to the direction from right to left in the above equivalence and requires that unchosen alternatives are chosen when all incident edges are reversed. \lrsout corresponds to the direction from left to right and requires that chosen alternatives are not chosen when all incident edges are reversed.
	
It follows directly from the definition that \lrsin (resp.~\lrsout) of a tournament solution $S$ carries over to any tournament solution that is a coarsening (resp.~refinement) of $S$.
	
\begin{lemma}\label{lem:lrs-inheritance}
	Let $S$ and $S'$ be two tournament solutions such that $S\subseteq S'$. If $S$ satisfies \lrsin, then so does $S'$. Conversely, if $S'$ satisfies \lrsout, then so does $S$.
\end{lemma}

There is an unexpected strong relationship between the purely choice-theoretic condition of stability and \lrs.
	
\begin{theorem}\label{thm:selfstable-lrsin}
	Every stable tournament solution satisfies \lrsin. 
\end{theorem}
	\begin{proof}\label{pf:}
		Suppose for contradiction that $S$ is stable but violates \lrsin. Then there exists a tournament $T=(A,\succ)$ and an alternative $a\in A$ such that $a\notin S(T)$ and $a\notin S(T^a)$.
		
		Let $T'=(A',\succ')$, where $A'=X\cup Y$ and each of $T'|_X$ and $T'|_Y$ is isomorphic to $T|_{A\setminus\{a\}}$. Also, partition $X=X_1\cup X_2$ and $Y=Y_1\cup Y_2$, where $X_1$ and $Y_1$ consist of the alternatives that are mapped to $\dom_T(a)$ by the isomorphism. To complete the definition of $T'$, we add the relations $X_1\succ'Y_2$, $Y_2\succ'X_2$, $X_2\succ'Y_1$, and $Y_1\succ'X_1$. The structure of tournament $T'$ is depicted in \figref{fig:selfstable-lrs}. 

We claim that both $X$ and $Y$ are externally $S$-stable in $T'$. To this end, we note that for every alternative $x\in X$ (resp. $y\in Y$) the subtournament $T|_{Y\cup\{x\}}$ (resp. $T|_{X\cup\{y\}}$) is isomorphic either to $T$ or to $T^a$, with $x$ (resp. $y$) being mapped to $a$. By assumption, $a$ is neither chosen in $T$ nor in $T^a$, and therefore $X$ and $Y$ are both externally $S$-stable in $T'$.

Now, suppose that $S(X\cup\{y\})=X'\subseteq X$ for some $y\in Y$. Note that $X'\neq\emptyset$ because tournament solutions always return non-empty sets. Since $S$ satisfies $\widehat\alpha$, we have $S(X)=X'$. Hence, $S(X)=X'=S(X\cup\{y\})$ for all $y\in Y$. Since $S$ satisfies $\widehat\gamma$, we also have $S(X\cup Y)=X'$. Similarly, we can deduce that $S(X\cup Y)=Y'$ for some $\emptyset\neq Y'\subseteq Y$. This yields the desired contradiction.
		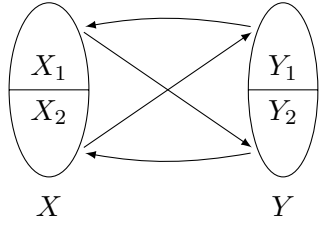
\begin{figure}[htbp]
		    \centering
		    \begin{tikzpicture}
		        [dom/.style={-latex, shorten >=1mm, shorten <=1mm}
		        ]
		        \node[draw, ellipse split, minimum height=6em] (T1){$X_1$ \nodepart{lower}$X_2$};
		        \node[draw, ellipse split, minimum height=6em, node distance=8em, right of=T1] (T2) {$Y_1$ \nodepart{lower} $Y_2$};
		        \node[below of=T1] {$X$};
		        \node[below of=T2] {$Y$};
		        \draw[dom] (T1.north east) to (T2.south west);
		        \draw[dom] (T2.south west) to[bend left=10] (T1.south east);
		        \draw[dom] (T1.south east) to (T2.north west);
		        \draw[dom] (T2.north west) to[bend right=10] (T1.north east);
		    \end{tikzpicture}
		    \caption{Construction of a tournament $T'$ with two $S$-stable sets $X$ and $Y$ used in the proof of \thmref{thm:selfstable-lrsin}.}
		    \label{fig:selfstable-lrs}
		\end{figure}
	\end{proof}
	
\subsection{Disproving Stability}
	
As discussed in \secref{sec:tsolutions}, disproving that a tournament solution satisfies stability can be very difficult. By virtue of \thmref{thm:selfstable-lrsin}, it now suffices to show that the tournament solution violates \lrsin. For $\teq$, this leads to the first concrete tournament in which $\teq$ violates stability. With the help of a computer, we found a minimal tournament in which \teq violates \lrsin using exhaustive search. This tournament is of order $13$ and thereby lies exactly at the boundary of the class of tournaments for which exhaustive search is possible. 
Using the construction explained in the proof of \thmref{thm:selfstable-lrsin}, we thus obtain a tournament of order $24$ in which $\teq$ violates $\widehat{\gamma}$. This tournament consists of two disjoint isomorphic subtournaments of order $12$ both of which are $\teq$-retentive. 

\begin{theorem}
\teq violates \lrsin.
\end{theorem}

\begin{proof}
We define a tournament $T=(\{x_1,\dots,x_{13}\},\succ)$ such that $x_{13}\not\in \teq(T)$ and $x_{13}\not\in \teq(T^{x_{13}})$.
The dominator sets in $T$ are defined as follows:
\[
\begin{array}{lcllcl}
	\dom(x_1)		&=&	\{x_4,x_5,x_6,x_8,x_9,x_{12}	\}\text{, }	&
	\dom(x_2)		&=&	\{x_1,x_6,x_7,x_{10},x_{12}		\}\text{, }\\	
	\dom(x_3)		&=&	\{x_1,x_2,x_6,x_7,x_9,x_{10}	\}\text{, }	&
	\dom(x_4)		&=&	\{x_2,x_3,x_7,x_8,x_{11}		\}\text{, }\\
	\dom(x_5)		&=&	\{x_2,x_3,x_4,x_8,x_{10},x_{11}	\}\text{, }	&
	\dom(x_6)		&=&	\{x_4,x_5,x_9,x_{11},x_{12}		\}\text{, }\\
	\dom(x_7)		&=&	\{x_1,x_5,x_6,x_{11},x_{12},x_{13}		\}\text{, }	&
	\dom(x_8)		&=&	\{x_2,x_3,x_6,x_7,x_{12},x_{13}		\}\text{, }\\	
	 \dom(x_9)	&=&	\{x_2,x_4,x_5,x_7,x_8,x_{13}			\}\text{, }	&
	 \dom(x_{10})	&=&	\{x_1,x_4,x_6,x_7,x_8,x_9,x_{13}		\}\text{, }\\
	 \dom(x_{11})	&=&	\{x_1,x_2,x_3,x_8,x_9,x_{10},x_{13}	\}\text{, }	&
	 \dom(x_{12})	&=&	\{x_3,x_4,x_5,x_9,x_{10},x_{11},x_{13}	\}\text{, }\\
	 \dom(x_{13}) &=& \{x_1,x_2,x_3,x_4,x_5,x_6\}\text{.}
\end{array}
\]

A rather tedious check reveals that
\[
\begin{array}{lcllcl}
	\teq(\dom(x_1))  	&=& \{x_4,x_8,x_{12}\}    \text{, } &
	\teq(\dom(x_2))  	&=&  \{x_6,x_{10},x_{12}\} \text{, }\\
	\teq(\dom(x_3))  	&=&  \{x_6,x_7,x_9\} 	    \text{, } &
	\teq(\dom(x_4))  	&=&  \{x_2,x_7,x_{11}\}    \text{, }\\
	\teq(\dom(x_5))  	&=&  \{x_2,x_8,x_{10}\}    \text{, } &
	\teq(\dom(x_6))  	&=&  \{x_4,x_9,x_{11}\}    \text{, }\\
	\teq(\dom(x_7))  	&=&  \{x_1,x_5,x_{11}\}    \text{, } &
	\teq(\dom(x_8))  	&=&  \{x_3,x_6,x_{12}\}    \text{, }\\
	\teq(\dom(x_9))  	&=&  \{x_2,x_5,x_{7}\}     \text{, } &
	\teq(\dom(x_{10}))&=& \{x_4,x_6,x_7\} 	    \text{, }\\
	\teq(\dom(x_{11}))&=& \{x_1,x_2,x_8\} 	    \text{, and } &
	\teq(\dom(x_{12}))&=& \{x_3,x_4,x_9\} 	    \text{.}\\
\end{array}
\]
It can then be checked that $\teq(T)=\teq(T^{x_{13}})=\{x_1,\dots,x_{12}\}$.
\end{proof}

Let $\nteq$ be the greatest natural number $n$ such that all tournaments of order $n$ or less admit a unique inclusion-minimal $\teq$-retentive set.
Together with earlier results by \citet{BFHM09a} and \citet{Yang16a}, we now have that $14 \leq \nteq \leq 23$.

The tournament used in the preceding proof does not show that $\me$ (or $\ba$) violate \lrsin. A computer search for such tournaments was unsuccessful. While it is known that $\me$ violates stability, a concrete counterexample thus remains elusive.

\subsection{Most Discriminating Stable Tournament Solutions}

An important property of tournament solutions that is not captured by any of the axioms introduced in \secref{sec:prelims} is discriminative power.\footnote{To see that discriminative power is not captured by the axioms, observe that the trivial tournament solution satisfies stability, monotonicity, regularity, and composition-consistency.} 
It is known that $\ba$ and $\mc$ (and by the known inclusions also $\uc$ and $\tc$) almost always select all alternatives when tournaments are drawn uniformly at random and the number of alternatives goes to infinity \citep{Fey08a,ScFe11a}.\footnote{However, these analytic results stand in sharp contrast to empirical observations that Condorcet winners are likely to exist in real-world settings, which implies that tournament solutions are much more discriminative than results for the uniform distribution suggest \citep{BrSe15a}.}
 Experimental results suggest that the same is true for $\teq$. Other tournament solutions, which are known to return small choice sets, fail to satisfy stability and composition-consistency. A challenging question is how discriminating tournament solutions can be while still satisfying desirable axioms.

\lrs reveals an illuminating dichotomy in this context for common tournament solutions. We state without proof that discriminating tournament solutions such as Copeland's rule, Slater's rule, and Markov's rule satisfy \lrsout and violate \lrsin. On the other hand, coarse tournament solutions such as \tc, \uc, and \mc satisfy \lrsin and violate \lrsout. The bipartisan set hits the sweet spot because it is the only one among the commonly considered tournament solutions that satisfies \lrsin \emph{and} \lrsout (and hence \lrs).

\begin{theorem}\label{thm:bp-lrs}
	\bp satisfies \lrs.
\end{theorem}
\begin{proof}
		Since \bp is stable, \thmref{thm:selfstable-lrsin} implies that $\bp$ satisfies \lrsin.
		Now, assume for contradiction that \bp violates \lrsout, \ie there is a tournament $T=(A,\succ)$ and an alternative $a$ such that $a \in \bp(T)$ and $a\in \bp(T^a)$. For a probability distribution~$p$ and a subset of alternatives $B\subseteq A$, let $p(B) = \sum_{x\in B} p(x)$. Consider the optimal mixed strategy $p_{T|_{A\setminus\{a\}}}$ in tournament $T|_{A\setminus\{a\}}$.
		It is known from \citet[][Prop.~6.4.8]{Lasl97a} that $a\in\bp(T)$ if and only if $p_{T|_{A\setminus\{a\}}}(D(a))>p_{T|_{A\setminus\{a\}}}(\dom(a))$. For $T^a$, we thus have that 
$p_{T^a|_{A\setminus\{a\}}}(D(a))>p_{T^a|_{A\setminus\{a\}}}(\dom(a))$. This is a contradiction because $D_T(a)=\dom_{T^a}(a)$ and $\dom_T(a)=D_{T^a}(a)$.
\end{proof}

The relationship between \lrs and the discriminative power of tournament solutions is no coincidence. To see this, consider all \emph{labeled} tournaments of fixed order and an arbitrary alternative $a$. \lrsin demands that $a$ is chosen in \emph{at least} one of $T$ and $T^a$ while \lrsout requires that $a$ is chosen in \emph{at most} one of $T$ and $T^a$. We thus obtain the following consequences.
	
\begin{itemize}
\item A tournament solution satisfying \lrsin chooses on average at least half of the alternatives.
\item A tournament solution satisfying \lrsout chooses on average at most half of the alternatives.
\item A tournament solution satisfying \lrs chooses on average half of the alternatives.
\end{itemize}

Hence, the well-known fact that \bp chooses on average half of the alternatives \citep{FiRe95a} follows from \thmref{thm:bp-lrs}.
Also, all coarsenings of $\bp$ such as \mc, \uc, and \tc satisfy \lrsin by virtue of \lemref{lem:lrs-inheritance}. On the other hand, since these tournament solutions are all different from \bp, they choose on average more than half of the alternatives and hence cannot satisfy \lrsout.

These results already hint at $\bp$ being perhaps a ``most discriminating'' stable tournament solution. In order to make this precise, we formally define the discriminative power of a tournament solution. For two tournament solutions $S$ and $S'$, we say that $S$ is \emph{more discriminating} than $S'$ if there is $n\in \mathbb{N}$ such that the average number of alternatives chosen by $S$ is lower than that of $S'$ over all labeled tournaments of order $n$. Note that this definition is very weak because we only have an existential, not a universal, quantifier for $n$. It is therefore even possible that two tournament solutions are more discriminating than each other. However, this only strengthens the following results.
Combining Theorems \ref{thm:selfstable-lrsin} and \ref{thm:bp-lrs} immediately yields the following theorem.

\begin{theorem}
\label{thm:stablelrs}
A stable tournament solution satisfies \lrs if and only if there is no more discriminating stable tournament solution.
\end{theorem}
\begin{proof}
First consider the direction from left to right. Let $S$ be a tournament solution that satisfies \lrs.
Due to the observations made above, $S$ chooses on average half of the alternatives. Since any stable tournament solution satisfies \lrsin by \thmref{thm:selfstable-lrsin}, it chooses on average at least half of the alternatives and therefore cannot be more discriminating than $S$.

Now consider the direction from right to left. Let $S$ be a most discriminating stable tournament solution.
Again, since any stable tournament solution satisfies \lrsin, $S$ chooses on average at least half of the alternatives. On the other hand, \bp is a stable tournament solution that chooses on average exactly half of the alternatives. This means that $S$ must also choose on average half of the alternatives, implying that it also satisfies \lrsout and hence \lrs.
\end{proof}

\begin{corollary}
\label{thm:nostablerefinementbp}
There is no more discriminating stable tournament solution than \bp. In particular, there is no stable refinement of \bp.
\end{corollary}

Given Corollary \ref{thm:nostablerefinementbp}, a natural question is whether every stable tournament solution that satisfies mild additional properties such as monotonicity is a coarsening of \bp. We give an example in Appendix \ref{app:s7hat} which shows that this is not true.

Finally, we provide two axiomatic characterizations of \bp by leveraging other traditional properties. These characterizations leverage the following lemma, which entails that, in order to show that two stable tournament solutions that satisfy \lrs are identical, it suffices to show that their roots are contained in each other.

\begin{lemma}
\label{lem:stablelrsequal}
Let $S$ and $S'$ be two stable tournament solutions satisfying \lrs. Then $\rootsym[S]\subseteq\rootsym[S']$ if and only if $S=S'$.
\end{lemma}

\begin{proof}
Suppose that $\rootsym[S]\subseteq\rootsym[S']$, and consider any tournament $T$. We will show that $S(T)\subseteq S'(T)$. If $S'(T)=T$, this is already the case. Otherwise, we have $S'(S'(T)\cup\{a\})=S'(T)$ for each $a\not\in S'(T)$. By definition of the root function, $\rootsym[S'](S'(T)\cup\{a\})=S'(T)$. Since the root function excludes at most one alternative from any tournament, we also have $\rootsym[S](S'(T)\cup\{a\})=S'(T)$ by our assumption $\rootsym[S]\subseteq\rootsym[S']$. Hence $S(S'(T)\cup\{a\})=S'(T)$ as well. Using $\widehat{\gamma}$ of $S$, we find that $S(T)=S'(T)$. So $S(T)\subseteq S'(T)$ for every tournament $T$. However, since $S$ and $S'$ satisfy \lrs, and therefore choose on average half of the alternatives, we must have $S=S'$.

Finally, if $S=S'$, then clearly $\rootsym[S]=\rootsym[S']$ and so $\rootsym[S]\subseteq\rootsym[S']$.
\end{proof}

\begin{theorem}
\label{thm:BPcharLRS}
\bp is the only tournament solution that satisfies stability, composition-consistency, monotonicity, regularity, and \lrs.
\end{theorem}

\begin{proof}
Let $S$ be a tournament solution satisfying the five aforementioned properties. Since $S$ and \bp are stable and satisfy \lrs, by \lemref{lem:stablelrsequal} it suffices to show that $\rootsym[S]\subseteq\rootsym[\bp]$. This is equivalent to showing that when $\rootsym[\bp]$ excludes an alternative from a tournament, then $\rootsym[S]$ excludes the same alternative. In other words, we need to show that when \bp excludes exactly one alternative $a$, then $S$ also only excludes $a$.

Let $T$ be a tournament in which \bp excludes exactly one alternative $a$. 
As defined in \secref{sec:bpandteq}, $\bp(T)$ corresponds to the support of the unique Nash equilibrium of $G(T)$. \citet{LLL93b} and \citet{FiRy95a} have shown that this support is always of odd size and that the equilibrium weights associated to the alternatives in $\bp(T)$ are odd numbers. 
Hence, using composition-consistency of \bp and the fact that the value of a symmetric zero-sum game is zero, $T$ can be transformed into a new (possibly larger) tournament $T_1=(A,\succ)$ by replacing each alternative except $a$ with a regular tournament of odd order such that $T_1|_{A\setminus\{a\}}$ is regular. Moreover, in $T_1$, $|\dom(a)|>\frac{|A|}{2}$. 

We will now show that $a\not\in S(T_1)$. Since $S$ is monotonic, it suffices to prove this when we strengthen $a$ arbitrarily against alternatives in $T_1$ until $|\dom(a)|=\frac{|A|+1}{2}$. 

Let $X=D(a)$ and $Y=\dom(a)$, and let $T_2$ be a tournament obtained by adding a new alternative $b$ to $T_1$ so that $X\succ \{b\}\succ Y$ and $a\succ b$. Note that $T_2$ is again a regular tournament, so $S(T_2)=A\cup \{b\}$. In particular, $b\in S(T_2)$.

	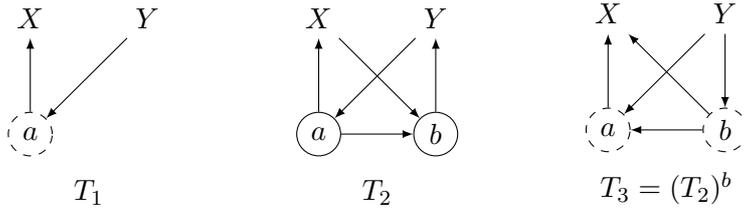
\begin{figure}[htb]
		\centering
		\begin{tikzpicture}[]
			\node (V) at (0,0) {$X$};
			\node (W) [right of=V] {$Y$};
			\node (y) [below of=W] {};
			\node[vertex,dashed] (x) [left of=y] {$a$};
			\foreach \x / \y in {x/V,W/x}{
			\draw[-latex] (\x) to (\y);
			};
			\node (caption) [node distance = 0.5*\nd,below of=d, xshift=0.5*\nd] {$T_1$};
		\end{tikzpicture}
		\qquad\qquad
		\begin{tikzpicture}[]
			\node (V) at (0,0) {$X$};
			\node (W) [right of=V] {$Y$};
			\node[vertex] (y) [below of=W] {$b$};
			\node[vertex] (x) [left of=y] {$a$};
			\foreach \x / \y in {x/V,x/y,y/W,V/y,W/x}{
			\draw[-latex] (\x) to (\y);
			};
			\node (caption) [node distance = 0.5*\nd,below of=d, xshift=0.5*\nd] {$T_2$};
		\end{tikzpicture}
		\qquad\qquad
		\begin{tikzpicture}[]
			\node (V) at (0,0) {$X$};
			\node (W) [right of=V] {$Y$};
			\node[vertex,dashed] (y) [below of=W] {$b$};
			\node[vertex,dashed] (x) [left of=y] {$a$};
			\foreach \x / \y in {x/V,y/x,y/V,W/x,W/y}{
			\draw[-latex] (\x) to (\y);
			};
			\node (caption) [node distance = 0.5*\nd,below of=d, xshift=0.5*\nd] {$T_3 = (T_2)^b$};
		\end{tikzpicture}
		\caption{Illustration of the proof of \thmref{thm:BPcharLRS}. Circled alternatives are contained in the choice set $S(\cdot)$. Alternatives circled with a dashed line are not contained in the choice set $S(\cdot)$.}
		\label{fig:BPproof}
	\end{figure}

Let $T_3=(T_2)^b$ be the tournament obtained from $T_2$ by reversing all edges incident to $b$. By \lrsout, we have $b\not\in S(T_3)$. If it were the case that $a\in S(T_3)$, then it should remain chosen when we reverse the edge between $a$ and $b$. However, alternative $a$ in the tournament after reversing the edge is isomorphic to alternative $b$ in $T_3$, and so we must have $b\in S(T_3)$, a contradiction. Hence $a\not\in S(T_3)$. Since $S$ satisfies stability and thus $\widehat{\alpha}$, we also have $a\not\in S(T_1)$, as claimed. See \figref{fig:BPproof} for an illustration.

Now, $\widehat{\alpha}$ and regularity of $S$ imply that $S(T_1)=S(T_1|_{A\backslash\{a\}})=A\backslash\{a\}$. Since $S$ satisfies composition-consistency, we also have that $S$ returns all alternatives except $a$ from the original tournament $T$, completing our proof.
\end{proof}

Based on Theorems \ref{thm:stablelrs} and \ref{thm:BPcharLRS}, we obtain another characterization that does not involve \lrs and hence only makes use of properties previously considered in the literature. 

\begin{corollary}
\label{thm:BPcharsize}
\bp is the unique most discriminating tournament solution that satisfies stability, composition-consistency, monotonicity, and regularity.
\end{corollary}

\begin{proof}
Suppose that a tournament solution $S$ satisfies stability, composition-consistency, monotonicity, and regularity and is as discriminating as \bp. By \thmref{thm:selfstable-lrsin}, $S$ satisfies \lrsin. Since $S$ chooses on average half of the alternatives, it satisfies \lrsout and hence \lrs as well. \thmref{thm:BPcharLRS} then implies that $S=BP$.
\end{proof}

The only previous characterization of \bp that we are aware of was obtained by \citet[][Thm.~6.3.10]{Lasl97a} and is based on a rather contrived property called \emph{Copeland-dominance}. According to \citet[][p.~153]{Lasl97a}, ``this axiomatization of the Bipartisan set does not add much to our knowledge of the concept because it is merely a re-statement of previous propositions.''
\coref{thm:BPcharsize} essentially shows that, for most discriminating stable tournament solutions, Laslier's Copeland-dominance is implied by monotonicity and regularity.

We now address the independence of the axioms used in \thmref{thm:BPcharLRS}.

\begin{remark}
\lrs is not implied by the other properties. In fact, the trivial tournament solution satisfies stability, composition-consistency, monotonicity, and regularity.
\end{remark}

\begin{remark}
Monotonicity is not implied by the other properties. In fact, the tournament solution that returns $\bp(\overline{T})$, where $\overline{T}$ is the tournament in which all edges in $T$ are reversed, satisfies stability, composition-consistency, regularity, and \lrs.
\end{remark}

The question of whether stability, composition-consistency, and regularity are independent in the presence of the other axioms is quite challenging. We can only provide the following partial answers.

\begin{remark}\label{rem:pos}
Neither stability nor composition-consistency is implied by \lrs, monotonicity, and regularity. In fact, there is a tournament solution that satisfies \lrs, monotonicity, and regularity, but violates stability and composition-consistency (see Appendix~\ref{app:pos}). 
\end{remark}

\begin{remark}\label{rem:s7hat}
Neither regularity nor composition-consistency is implied by stability and monotonicity. In fact, there is a tournament solution that satisfies stability and monotonicity, but violates regularity and composition-consistency (see Appendix~\ref{app:s7hat}).
\end{remark}

\citet{BHS15a} brought up the question whether stability implies regularity (under mild assumptions) because all stable tournaments solutions studied prior to this paper were regular.\footnote{We checked on a computer that the stable tournament solution \tcring \citep[see][]{BBFH11a,BBH15a} satisfies regularity for all tournaments of order 17 or less.} \remref{rem:s7hat} shows that this does not hold without making assumptions that go beyond monotonicity.

Given the previous remarks, it is possible that composition-consistency and regularity are not required for \thmref{thm:BPcharLRS} and \coref{thm:BPcharsize}.
Indeed, our computer experiments have shown that the only stable and monotonic tournament solution satisfying \lrs for all tournaments of order up to 7 is \bp. This may, however, be due to the large number of automorphisms in small tournaments and composition-consistency and regularity could be required for larger tournaments.
It is also noteworthy that the proof of \thmref{thm:BPcharLRS} only requires a weak version of composition-consistency, where 
all components are tournaments in which all alternatives are returned due to automorphisms.

Since stability is implied by Samuelson's weak axiom of revealed preference or, equivalently, by transitive rationalizability, \coref{thm:BPcharsize} can be seen as an escape from Arrow's impossibility theorem where the impossibility is turned into a possibility by weakening transitive rationalizability and (significantly) strengthening the remaining conditions \citep[see, also,][]{BrHa11a}.

\section*{Acknowledgements}
This material is based on work supported by Deutsche Forschungsgemeinschaft under grants {BR~2312/7-1} and {BR~2312/7-2}, by a Feodor Lynen Research Fellowship of the Alexander von Humboldt Foundation, by ERC Starting Grant 639945, by a Stanford Graduate Fellowship, and by the MIT-Germany program.
The authors thank Christian Geist for insightful computer experiments and Paul Harrenstein for helpful discussions and preparing \figref{fig:stability-illustration}.

\bibliography{abb,brandt,pamas}

\begin{thebibliography}{43}
\providecommand{\natexlab}[1]{#1}
\providecommand{\url}[1]{\texttt{#1}}
\expandafter\ifx\csname urlstyle\endcsname\relax
  \providecommand{\doi}[1]{doi: #1}\else
  \providecommand{\doi}{doi: \begingroup \urlstyle{rm}\Url}\fi

\bibitem[Aizerman and Aleskerov(1995)]{AiAl95a}
M.~Aizerman and F.~Aleskerov.
\newblock \emph{Theory of Choice}, volume~38 of \emph{Studies in Mathematical
  and Managerial Economics}.
\newblock North-Holland, 1995.

\bibitem[Allesina and Levine(2011)]{AlLe11a}
S.~Allesina and J.~M. Levine.
\newblock A competitive network theory of species diversity.
\newblock \emph{Proceedings of the National Academy of Sciences (PNAS)},
  108\penalty0 (14):\penalty0 5638--5642, 2011.

\bibitem[Bordes(1979)]{Bord79a}
G.~Bordes.
\newblock Some more results on consistency, rationality and collective choice.
\newblock In J.~J. Laffont, editor, \emph{Aggregation and Revelation of
  Preferences}, chapter~10, pages 175--197. North-Holland, 1979.

\bibitem[Brandl et~al.(2016)Brandl, Brandt, and Seedig]{Bran13a}
F.~Brandl, F.~Brandt, and H.~G. Seedig.
\newblock Consistent probabilistic social choice.
\newblock \emph{Econometrica}, 84\penalty0 (5):\penalty0 1839--1880, 2016.

\bibitem[Brandt(2011)]{Bran11b}
F.~Brandt.
\newblock Minimal stable sets in tournaments.
\newblock \emph{Journal of Economic Theory}, 146\penalty0 (4):\penalty0
  1481--1499, 2011.

\bibitem[Brandt(2015)]{Bran11c}
F.~Brandt.
\newblock Set-monotonicity implies {K}elly-strategyproofness.
\newblock \emph{Social Choice and Welfare}, 45\penalty0 (4):\penalty0 793--804,
  2015.

\bibitem[Brandt and Harrenstein(2011)]{BrHa11a}
F.~Brandt and P.~Harrenstein.
\newblock Set-rationalizable choice and self-stability.
\newblock \emph{Journal of Economic Theory}, 146\penalty0 (4):\penalty0
  1721--1731, 2011.

\bibitem[Brandt and Seedig(2016)]{BrSe15a}
F.~Brandt and H.~G. Seedig.
\newblock On the discriminative power of tournament solutions.
\newblock In \emph{Selected Papers of the International Conference on
  Operations Research, OR2014}, Operations Research Proceedings, pages 53--58.
  Springer-Verlag, 2016.

\bibitem[Brandt et~al.(2010)Brandt, Fischer, Harrenstein, and Mair]{BFHM09a}
F.~Brandt, F.~Fischer, P.~Harrenstein, and M.~Mair.
\newblock A computational analysis of the tournament equilibrium set.
\newblock \emph{Social Choice and Welfare}, 34\penalty0 (4):\penalty0 597--609,
  2010.

\bibitem[Brandt et~al.(2013)Brandt, Chudnovsky, Kim, Liu, Norin, Scott,
  Seymour, and Thomass\'{e}]{BCK+11a}
F.~Brandt, M.~Chudnovsky, I.~Kim, G.~Liu, S.~Norin, A.~Scott, P.~Seymour, and
  S.~Thomass\'{e}.
\newblock A counterexample to a conjecture of {S}chwartz.
\newblock \emph{Social Choice and Welfare}, 40\penalty0 (3):\penalty0 739--743,
  2013.

\bibitem[Brandt et~al.(2014)Brandt, Brill, Fischer, and Harrenstein]{BBFH11a}
F.~Brandt, M.~Brill, F.~Fischer, and P.~Harrenstein.
\newblock Minimal retentive sets in tournaments.
\newblock \emph{Social Choice and Welfare}, 42\penalty0 (3):\penalty0 551--574,
  2014.

\bibitem[Brandt et~al.(2015)Brandt, Dau, and Seedig]{BDS13a}
F.~Brandt, A.~Dau, and H.~G. Seedig.
\newblock Bounds on the disparity and separation of tournament solutions.
\newblock \emph{Discrete Applied Mathematics}, 187:\penalty0 41--49, 2015.

\bibitem[Brandt et~al.(2016)Brandt, Brill, and Harrenstein]{BBH15a}
F.~Brandt, M.~Brill, and P.~Harrenstein.
\newblock Tournament solutions.
\newblock In F.~Brandt, V.~Conitzer, U.~Endriss, J.~Lang, and A.~D. Procaccia,
  editors, \emph{Handbook of Computational Social Choice}, chapter~3. Cambridge
  University Press, 2016.

\bibitem[Brandt et~al.(2017)Brandt, Harrenstein, and Seedig]{BHS15a}
F.~Brandt, P.~Harrenstein, and H.~G. Seedig.
\newblock Minimal extending sets in tournaments.
\newblock \emph{Mathematical Social Sciences}, 87:\penalty0 55--63, 2017.

\bibitem[Chernoff(1954)]{Cher54a}
H.~Chernoff.
\newblock Rational selection of decision functions.
\newblock \emph{Econometrica}, 22\penalty0 (4):\penalty0 422--443, 1954.

\bibitem[Duddy et~al.(2014)Duddy, Houy, Lang, Piggins, and Zwicker]{DHLP+14a}
C.~Duddy, N.~Houy, J.~Lang, A.~Piggins, and W.~S. Zwicker.
\newblock Social dichotomy functions.
\newblock Working paper, 2014.

\bibitem[Dutta(1988)]{Dutt88a}
B.~Dutta.
\newblock Covering sets and a new {C}ondorcet choice correspondence.
\newblock \emph{Journal of Economic Theory}, 44\penalty0 (1):\penalty0 63--80,
  1988.

\bibitem[Felsenthal and Machover(1992)]{FeMa92a}
D.~S. Felsenthal and M.~Machover.
\newblock After two centuries should {C}ondorcet's voting procedure be
  implemented?
\newblock \emph{Behavioral Science}, 37\penalty0 (4):\penalty0 250--274, 1992.

\bibitem[Fey(2008)]{Fey08a}
M.~Fey.
\newblock Choosing from a large tournament.
\newblock \emph{Social Choice and Welfare}, 31\penalty0 (2):\penalty0 301--309,
  2008.

\bibitem[Fishburn(1984)]{Fish84a}
P.~C. Fishburn.
\newblock Probabilistic social choice based on simple voting comparisons.
\newblock \emph{Review of Economic Studies}, 51\penalty0 (4):\penalty0
  683--692, 1984.

\bibitem[Fisher and Reeves(1995)]{FiRe95a}
D.~C. Fisher and R.~B. Reeves.
\newblock Optimal strategies for random tournament games.
\newblock \emph{Linear Algebra and its Applications}, 217:\penalty0 83--85,
  1995.

\bibitem[Fisher and Ryan(1995)]{FiRy95a}
D.~C. Fisher and J.~Ryan.
\newblock Tournament games and positive tournaments.
\newblock \emph{Journal of Graph Theory}, 19\penalty0 (2):\penalty0 217--236,
  1995.

\bibitem[Houy(2009{\natexlab{a}})]{Houy09a}
N.~Houy.
\newblock Still more on the tournament equilibrium set.
\newblock \emph{Social Choice and Welfare}, 32:\penalty0 93--99,
  2009{\natexlab{a}}.

\bibitem[Houy(2009{\natexlab{b}})]{Houy09b}
N.~Houy.
\newblock A few new results on {TEQ}.
\newblock {M}imeo, 2009{\natexlab{b}}.

\bibitem[Hudry(2009)]{Hudr09a}
O.~Hudry.
\newblock A survey on the complexity of tournament solutions.
\newblock \emph{Mathematical Social Sciences}, 57\penalty0 (3):\penalty0
  292--303, 2009.

\bibitem[Kreweras(1965)]{Krew65a}
G.~Kreweras.
\newblock Aggregation of preference orderings.
\newblock In \emph{Mathematics and Social Sciences {I}: Proceedings of the
  seminars of {Menthon-Saint-Bernard}, {F}rance (1--27 July 1960) and of
  G{\"o}sing, {A}ustria (3--27 July 1962)}, pages 73--79, 1965.

\bibitem[Laffond et~al.(1993{\natexlab{a}})Laffond, Laslier, and {Le
  Breton}]{LLL93a}
G.~Laffond, J.-F. Laslier, and M.~{Le Breton}.
\newblock More on the tournament equilibrium set.
\newblock \emph{Math{\'e}matiques et sciences humaines}, 31\penalty0
  (123):\penalty0 37--44, 1993{\natexlab{a}}.

\bibitem[Laffond et~al.(1993{\natexlab{b}})Laffond, Laslier, and {Le
  Breton}]{LLL93b}
G.~Laffond, J.-F. Laslier, and M.~{Le Breton}.
\newblock The bipartisan set of a tournament game.
\newblock \emph{Games and Economic Behavior}, 5\penalty0 (1):\penalty0
  182--201, 1993{\natexlab{b}}.

\bibitem[Laffond et~al.(1994)Laffond, Laslier, and {Le Breton}]{LLL94c}
G.~Laffond, J.-F. Laslier, and M.~{Le Breton}.
\newblock The {Copeland} measure of {Condorcet} choice functions.
\newblock \emph{Discrete Applied Mathematics}, 55\penalty0 (3):\penalty0
  273--279, 1994.

\bibitem[Landau(1951)]{Land51a}
H.~G. Landau.
\newblock On dominance relations and the structure of animal societies:
  {I}.~{E}ffect of inherent characteristics.
\newblock \emph{Bulletin of Mathematical Biophysics}, 13\penalty0 (1):\penalty0
  1--19, 1951.

\bibitem[Laslier(1997)]{Lasl97a}
J.-F. Laslier.
\newblock \emph{Tournament Solutions and Majority Voting}.
\newblock Springer-Verlag, 1997.

\bibitem[Masatlioglu et~al.(2012)Masatlioglu, Nakajima, and Ozbay]{MNO12a}
Y.~Masatlioglu, D.~Nakajima, and E.~Y. Ozbay.
\newblock Revealed attention.
\newblock \emph{American Economic Review}, 102\penalty0 (5):\penalty0
  2183--2205, 2012.

\bibitem[McGarvey(1953)]{McGa53a}
D.~C. McGarvey.
\newblock A theorem on the construction of voting paradoxes.
\newblock \emph{Econometrica}, 21\penalty0 (4):\penalty0 608--610, 1953.

\bibitem[Mnich et~al.(2015)Mnich, Shrestha, and Yang]{MSY15a}
M.~Mnich, Y.~R. Shrestha, and Y.~Yang.
\newblock When does {S}chwartz conjecture hold?
\newblock In \emph{Proceedings of the 24th International Joint Conference on
  Artificial Intelligence (IJCAI)}, pages 603--609. AAAI Press, 2015.

\bibitem[Monjardet(2008)]{Monj08a}
B.~Monjardet.
\newblock Statement of precedence and a comment on {IIA} terminology.
\newblock \emph{Games and Economic Behavior}, 62:\penalty0 736--738, 2008.

\bibitem[Moser(2015)]{Mose15a}
S.~Moser.
\newblock Majority rule and tournament solutions.
\newblock In J.~C. Heckelman and N.~R. Miller, editors, \emph{Handbook of
  Social Choice and Voting}, chapter~6, pages 83--101. Edgar Elgar, 2015.

\bibitem[Saari and Barney(2003)]{SaBa03a}
D.~G. Saari and S.~Barney.
\newblock Consequences of reversing preferences.
\newblock \emph{Mathematical Intelligencer}, 25:\penalty0 17--31, 2003.

\bibitem[Schjelderup-Ebbe(1922)]{Schj22a}
T.~Schjelderup-Ebbe.
\newblock Beitr{\"a}ge zur {S}ozialpsychologie des {H}aushuhns.
\newblock \emph{Zeitschrift f{\"u}r Psychologie}, 88:\penalty0 225--252, 1922.

\bibitem[Schwartz(1990)]{Schw90a}
T.~Schwartz.
\newblock Cyclic tournaments and cooperative majority voting: {A} solution.
\newblock \emph{Social Choice and Welfare}, 7\penalty0 (1):\penalty0 19--29,
  1990.

\bibitem[Scott and Fey(2012)]{ScFe11a}
A.~Scott and M.~Fey.
\newblock The minimal covering set in large tournaments.
\newblock \emph{Social Choice and Welfare}, 38\penalty0 (1):\penalty0 1--9,
  2012.

\bibitem[Sen(1971)]{Sen71a}
A.~K. Sen.
\newblock Choice functions and revealed preference.
\newblock \emph{Review of Economic Studies}, 38\penalty0 (3):\penalty0
  307--317, 1971.

\bibitem[Slater(1961)]{Slat61a}
P.~Slater.
\newblock Inconsistencies in a schedule of paired comparisons.
\newblock \emph{Biometrika}, 48\penalty0 (3--4):\penalty0 303--312, 1961.

\bibitem[Yang(2016)]{Yang16a}
Y.~Yang.
\newblock A further step towards an understanding of the tournament equilibrium
  set.
\newblock Technical report, http://arxiv.org/abs/1611.03991, 2016.

\end{thebibliography}
\appendix

\section{Appendix}

\subsection{Examples for \remref{rem:alphagamma}}
\label{app:alphagamma}

\citet[][p.~1729]{BrHa11a} mention that $\widehat{\alpha}$ and $\widehat{\gamma}$ are independent from each other in the context of general choice functions. Here, we prove that the same holds even in the context of tournament solutions.

\begin{proposition}\label{prop:alpha-not-gamma}
	There exists a tournament solution that satisfies $\widehat\alpha$, but not $\widehat\gamma$. 
\end{proposition}

\begin{proof}Let $S$ be a stable tournament solution.
As mentioned in \secref{sec:generators}, $\rootsym[S]$ satisfies $\widehat{\alpha}$. However, it is easily seen that, unless $S$ is trivial, $\rootsym[S]$ violates $\widehat{\gamma}$.
Hence, the statement follows from the existence of non-trivial stable tournament solutions (such as \bp).
\end{proof}

\begin{proposition}\label{prop:gamma-not-alpha}
	There exists a tournament solution that satisfies $\widehat\gamma$, but not $\widehat\alpha$. 
\end{proposition}

\begin{proof}
	Let $S$ be a stable tournament solution. Define the tournament solution $S'$ such that for each tournament $T=(A,\succ)$,
	\[ S'(T) = 
	\begin{cases}
		S(T)  &\text{if } |A \setminus S(T)|>1 \\
		A      &\text{otherwise.}
	\end{cases}
	\]
	It can be shown that $S'$ satisfies $\widehat\gamma$, but may violate $\widehat\alpha$. For the latter, let $S=\tc$ and consider a transitive tournament $(\set{a,b,c},\succ)$ such that $a \succ b$, $b \succ c$, and $a \succ c$. By definition, $S'(\set{a,b,c})=\set{a}$, but $S'(\set{a,b})=\set{a,b}$.
\end{proof}
	
\subsection{Examples for \remref{rem:localalpha}}
\label{app:localalpha}

$BA$ satisfies local $\widehat{\alpha}$, but $\widehat{BA}=ME$ violates $\widehat{\alpha}$ \citep{BHS15a}. 

Similarly, there exists a tournament solution $S$ for which $\widehat{S}$ is well-defined, but $\widehat{S}$ is not stable.
For a stable tournament solution $S$, we have by definition that $S=\widehat{S}$ and hence that $\widehat{S}$ is also stable. The following proposition shows that $\widehat{\alpha}$ does not carry over from $S$ to $\widehat{S}$ even if $S$ is simple and $\widehat{S}$ is well-defined.

\begin{proposition}
\label{prop:examplenondirectedshat}
There exists a simple tournament solution $S$ satisfying $\widehat{\alpha}$ such that $\widehat{S}$ is well-defined but $\widehat{S}$ does not satisfy $\widehat{\alpha}$.
\end{proposition}

\begin{proof}
Let $S$ be the tournament solution that always chooses all alternatives, with two exceptions:
\begin{itemize}
\item If the tournament is of order 2, then $S$ chooses only the Condorcet winner.
\item If the tournament is the tournament $T_4$ given in \figref{fig:secretary}, then $S$ chooses alternatives $a, b,$ and $c$.
\end{itemize}

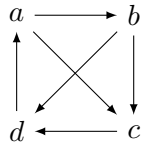
\begin{figure}[htb] 
	\centering
		\begin{tikzpicture}[]
			\node (a) at (0,0) {$a$};
			\node (b) [right of=a] {$b$};
			\node (c) [below of=b] {$c$};
			\node (d) [left of=c] {$d$};
			\foreach \x / \y in {a/b,a/c,b/c,b/d,c/d,d/a}{
			\draw[-latex] (\x) to (\y);
			};
		\end{tikzpicture}

	\caption{Tournament $T_4$.}
	\label{fig:secretary}
\end{figure}

Clearly, $S$ is simple and satisfies $\widehat{\alpha}$. Since $\widehat{S}$ chooses alternatives $a,b,$ and $c$ from $T_4$, but chooses only the Condorcet winner from the transitive tournament of order 3, it does not satisfy $\widehat{\alpha}$.

It remains to show that $\widehat{S}$ is well-defined. One can check that every tournament contains an $S$-stable set. Suppose for contradiction that some tournament $T$ contains two distinct minimal $S$-stable sets, which we denote by $B$ and $C$. Then $B$ and $C$ are also $S$-stable in $B\cup C$. If $B$ is a singleton, then $B$ is the Condorcet winner in $B\cup C$, which means $C$ cannot be $S$-stable, a contradiction. Hence both $B$ and $C$ are transitive tournaments of order 3, and $4\leq|B\cup C|\leq 6$. One can check all the possibilities of $B\cup C$ to conclude that this case is also impossible.
\end{proof}

For the tournament solution $S$ defined in the proof of Proposition \ref{prop:examplenondirectedshat}, we have that $\widehat{\widehat{S}}$ is not well-defined. Even though no tournament contains two distinct minimal $\widehat{S}$-stable sets, $T_4$ does not contain any $\widehat{S}$-stable set. This example also shows that for a tournament solution $S'$, $\widehat{S'}$ may fail to be well-defined not because it allows two distinct minimal $S'$-stable sets in a tournament but because some tournament contains no $S'$-stable set.

\subsection{Example for \remref{rem:pos}}
\label{app:pos}

We show that there exists a tournament solution different from \bp that satisfies \lrs, monotonicity, regularity, and Condorcet consistency.

To this end, we define a new tournament solution called \pos which chooses all alternatives with positive relative degree. More precisely, an alternative is chosen by \pos if it dominates strictly more than half of the remaining alternatives, is not chosen if it dominates strictly less than half of the remaining alternatives, and goes to a ``tie-break'' to determine whether it is chosen if it dominates exactly half of the remaining alternatives. 

For tournaments of even size, \pos chooses exactly the alternatives that dominate at least (or equivalently, more than) half of the remaining alternatives. Hence we do not need a tie-break for tournaments of even size. The tie-breaking rule for tournaments of odd size $2n+1$ is as follows: For any (unlabeled) tournament $T$ of order $2n$ and any partition of it into two sets $B$ and $C$ of size $n$, consider two tournaments $T_1$ and $T_2$ of order $2n+1$. The tournament $T_1$ contains $T$ and another alternative $a$ that dominates $B$ but is dominated by $C$, while the tournament $T_2$ contains $T$ and another alternative $a$ that dominates $C$ but is dominated by $B$. If $T_1$ or $T_2$ is regular, \pos chooses $a$ in that tournament and not in the other one. Otherwise, \pos arbitrarily chooses $a$ in exactly one of $T_1$ and $T_2$.

\begin{proposition}
\label{prop:poslrs}
\pos satisfies \lrs, monotonicity, and regularity.
\end{proposition}

\begin{proof}
We need to show that the tie-breaking rule in the definition of \pos is well-defined. First, we show that if we perform a local reversal on alternative $a$, we do not get an isomorphic tournament with alternative $a$ mapped to itself. Indeed, if $a$ were mapped to itself, it would mean that no tournament solution satisfies \lrs, which we know is not true since \bp satisfies \lrs. Secondly, the tournament obtained by performing a local reversal on an alternative in a regular tournament is not regular. Hence we do not obtain a conflict within the tie-breaking rule.

It follows directly from the definition that \pos satisfies \lrs, monotonicity, and regularity.
\end{proof}

The tournament $T_4$ given in \figref{fig:secretary} shows that \pos violates composition-consistency and $\widehat{\alpha}$ (and hence stability).

Interestingly, \bp (and all of its coarsenings) always intersect with \pos while there exists a tournament for which \ba (and all of its refinements such as \teq and \me) do \emph{not} overlap with \pos. This follows from results on the \emph{Copeland value} by \citet{LLL93b,LLL94c}.

\subsection{Example for \remref{rem:s7hat}}
\label{app:s7hat}

We construct a tournament solution that satisfies monotonicity and stability, but violates regularity and composition-consistency.

Every tournament solution has to be regular on tournaments of order 5 or less because of non-trivial automorphisms.
Consider the tournament $T_7$ shown in Figure \ref{fig:R7}. $T_7$ admits a unique nontrivial automorphism that maps each of the six alternatives in the two 3-cycles to the next alternative in its 3-cycle and maps alternative $g$ to itself.\footnote{Note that $(T_7)^g$ is the smallest tournament in which $BA$ and $UC$ differ \citep{BDS13a}.}


	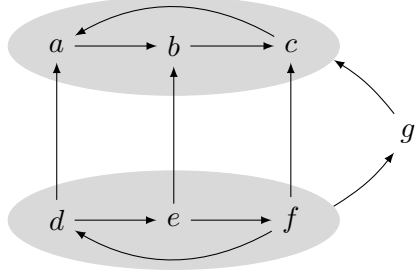
\begin{figure}
	\centering
	\tikzstyle{every ellipse node}=[draw,inner xsep=4em,inner ysep=1.2em,fill=black!15!white,draw=black!15!white]
	\begin{tikzpicture}
	 	\draw (0,1.5*\nd) node[ellipse] (ellipse1){}		++(-\nd,0)	node(a){$a$} ++(\nd,0) 	node(b){$b$} ++(\nd,0) node(c){$c$} ;
		\draw (0,0) node[ellipse] (ellipse2){}		++(-\nd,0)	node(d){$d$} ++(\nd,0) 	node(e){$e$} ++(\nd,0) node(f){$f$} ;
		\node (g) at ($ (c) !.5! (f)$)[xshift=1*\nd] {$g$};
		\draw [-latex] (ellipse2.5) to [bend right=10] (g);
		\draw [-latex] (g) to [bend right=10] (ellipse1.355);
		\draw [-latex] (a) to (b);
		\draw [-latex] (b) to (c);
		\draw [-latex] (d) to (e);
		\draw [-latex] (e) to (f);
		\draw [-latex] (d) to (a);
		\draw [-latex] (e) to (b);
		\draw [-latex] (f) to (c);
		\draw [-latex] (c) to[bend right=30] (a);
		\draw [-latex] (f) to[bend left=30] (d);
	\end{tikzpicture}
	\caption{Tournament $T_7$. $g\succ\{a,b,c\}$, $\{d,e,f\}\succ g$, and all omitted edges point downwards.}
	\label{fig:R7}
	\end{figure}

Now, define the simple tournament solution $S_7$, which always returns all alternatives unless the tournament is $T_7$ or it can be modified from $T_7$ by weakening alternative $g$. In the latter case, $S_7$ returns all alternatives except $g$.

We check that this definition is sound. First, we know that in $T_7$, there is no automorphism that maps alternative $g$ to another alternative. When we weaken $g$, it is the unique alternative with the smallest out-degree, and hence cannot be mapped by an automorphism to another alternative. Now, the alternatives $a,b,c$ form an orbit, and $S_7$ excludes $g$ whenever it is dominated by $d,e,f$ (and has any dominance relationship to $a,b,c$). This yields four non-isomorphic tournaments for which $S_7$ excludes $g$.

\begin{proposition}
$\widehat{S_7}$ satisfies stability and monotonicity.
\end{proposition}

\begin{proof}
First, observe that $S_7$ trivially satisfies local $\widehat{\alpha}$ because $S_7$ only excludes an alternative in tournaments of order 7. By virtue of \thmref{thm:ShatdirectedMSSP}, it therefore suffices to show that $\widehat{S_7}$ is well-defined.
	
One can check that every tournament contains an $S_7$-stable set. Let $T_6$ denote the tournament obtained by removing alternative $g$ from $T_7$. Suppose for contradiction that there exists a tournament $T$ that contains two distinct minimal $S_7$-stable sets, which we denote by $B$ and $C$. Then $B$ and $C$ are also $S_7$-stable in $B\cup C$. Moreover, $T|_{B}$ must correspond to the tournament $T_6$, and each alternative in $C\backslash B$ either has the same dominance relation to $B$ as the alternative $g$ does to $T_6$ or has a dominance relation that is a weakening of $g$. The same statement holds for $C$. We consider the following cases.

\textit{Case 1}: $10\leq |B\cup C|\leq 11$. The tournament $T|_{B}$ has one of its alternatives corresponding to alternatives $d$, $e$, and $f$ in \figref{fig:R7} outside of $B\cap C$, and this alternative must dominate all of the alternatives in $C$. Similarly, there exists an alternative in $C\backslash B$ that dominates all of the alternatives in $B$. But this implies that some two alternatives dominate each other, a contradiction.

\textit{Case 2}: $7\leq |B\cup C|\leq 9$. At least one of the two tournaments $T|_{B}$ and $T|_{C}$ must have all of its alternatives corresponding to alternatives $d$, $e$, and $f$ in \figref{fig:R7} in the intersection $B\cap C$, for otherwise we obtain a contradiction in the same way as in Case 1. Assume without loss of generality that $T|_{B}$ has its alternatives corresponding to $d$, $e$, and $f$ in the intersection. Hence three alternatives in $B\cap C$ that form a cycle dominate the same alternative in $C$. But this does not occur in $T_6$, a contradiction.

It follows from \thmref{thm:mon} that $\widehat{S_7}$ satisfies monotonicity.
\end{proof}

Clearly, $\widehat{S_7}$ is not regular since it excludes an alternative from the regular tournament $T_7$. Moreover, it is not a coarsening of \bp since \bp selects all of the alternatives in $T_7$. Hence we have that stable and monotonic tournament solutions are not necessarily coarsenings of \bp.

\end{document}